\documentclass[final] {IEEEtran}

\ifCLASSOPTIONonecolumn
\usepackage[margin=1.0in]{geometry} 
\else
\fi

\usepackage{cite} 
\usepackage{xspace}
\usepackage{filecontents}

\usepackage{float}
\usepackage{pgf, tikz}

\usepackage{graphicx}

\ifCLASSINFOpdf
\else
\fi
\usepackage[cmex10]{amsmath}

\usepackage{bm}

\usepackage{amssymb} 
\usepackage{array}

\usepackage{subfigure}
\usepackage{caption2} 
\usepackage{amsfonts}
\usepackage{amsthm}   

\usepackage{xcolor,multicol}
\usepackage{soul}

\newtheorem{theorem}{Theorem} 
\newtheorem{lemma}{Lemma}
\newtheorem{proposition}{Proposition}
\newtheorem{corollary}{Corollary}[proposition]

\usetikzlibrary{patterns,snakes}

\newcommand{\myexpect}[1]{\mathbb{E}\left\lbrace #1 \right\rbrace}

\newcommand{\AuthorOne}{Wanchun Liu}
\newcommand{\AuthorTwo}{Xiangyun Zhou}
\newcommand{\AuthorThree}{Salman Durrani}
\newcommand{\AuthorFour}{Petar Popovski}
\newcommand{\ThankOne}{Wanchun Liu, Xiangyun Zhou, Salman Durrani are with the Research School of Engineering, the
	Australian National University, Canberra, ACT 2601, Australia 
	(emails: \{wanchun.liu, xiangyun.zhou, salman.durrani\}@anu.edu.au). 
	Petar Popovski is with the Department of Electronic Systems, Aalborg University, Aalborg 9220, Denmark (email: petarp@es.aau.dk). 
The work of X. Zhou was supported by the Australian Research Council's Discovery Project funding scheme (project number DP150103905).
	}

\usepackage[euler]{textgreek}

\usepackage{textcomp}

\usepackage{enumerate}

\usepackage{color}

\begin{document}
%



%
%
%


\title{Secure Communication with a Wireless-Powered Friendly Jammer}
\author{\authorblockN{\AuthorOne,~\AuthorTwo,~\AuthorThree, and~\AuthorFour\thanks{\ThankOne}}}
\maketitle

\setcounter{page}{1}

\begin{abstract}
	In this paper, we propose to use a wireless-powered friendly jammer to enable secure communication between a source node and destination node, in the presence of an eavesdropper. 
	We consider a two-phase communication protocol with fixed-rate transmission. 
	In the first phase, wireless power transfer is conducted from the source to the jammer. 
	In the second phase, the source transmits the information-bearing signal under the protection of a jamming signal sent by the jammer using the harvested energy in the first phase. 
	We analytically characterize the long-term behavior of the proposed protocol and derive a closed-form expression for the throughput. 
	We further optimize the rate parameters for maximizing the throughput subject to a secrecy outage probability constraint. 
	Our analytical results show that the throughput performance differs significantly between the single-antenna jammer case and the multi-antenna jammer case. 
	For instance, as the source transmit power increases, 
	the throughput quickly reaches an upper bound with single-antenna jammer,
	while the throughput grows unbounded with multi-antenna jammer. 
	Our numerical results also validate the derived analytical results. 
\end{abstract}
%



\section{Introduction}
\subsection{Background and Motivation}
Physical layer security has been recently proposed as a complement to cryptography 
method to provide secure wireless communications \cite{Bookbloch,Bookzhou}. 
It is a very different paradigm where secrecy is achieved by exploiting the physical layer properties of the wireless communication system, especially interference and fading.
Several important physical layer security techniques have been investigated in the past decade (see a survey article \cite{survey_phy} and the references therein).
Inspired by cooperative communication without secrecy constraints, user cooperation is a promising strategy for improving secrecy performance. 
There are mainly two kinds of cooperation: cooperative relaying and cooperative jamming.
As for cooperative relaying, the well-known decode-and-forward and amplify-and-forward schemes were discussed in \cite{dongAF,vaneet,Jiangyuan} with secrecy considerations.
Following the idea of artificial noise in \cite{Negi},
cooperative jamming was investigated as an effective method to enhance secrecy \cite{Yener,Petar,Krikidis,gan,Rongqing,dong,huang,another_ali,JunYang}.
In this scheme, a friendly jammer transmits a jamming signal 
to interfere with the eavesdropper's signal reception 
at the same time when the source transmits the message signal.
In \cite{Yener,Petar,Krikidis}, the authors focused on the design of a single-antenna jammer. 
In \cite{gan} and \cite{Rongqing}, multiple single-antenna jammers were considered to generate distributed cooperative jamming signals. 
In \cite{dong}, the authors studied multi-antenna jammer (called relay in \cite{dong}) in secure wireless networks. 
Motivated by this work, the authors in \cite{huang,another_ali,JunYang} considered multi-antenna jammers in MIMO (multiple-input and multiple-output) networks.
\par
In many wireless network applications, communication nodes may not have connection to power lines due to mobility or other constraints. Thus, their lifetime is usually constrained by energy stored in the battery.
In order to prolong the lifetime of energy-constrained wireless nodes, 
 energy harvesting has been proposed as a very promising approach \cite{Varshney,Grover}. 
 {Conventional
 	energy harvesting methods rely on various renewable energy sources in the environment, such
 	as solar, vibration, thermoelectric and wind, thus are usually uncontrollable.}
%
For a wireless communication environment, harvesting energy from radio-frequency (RF) signals has recently attracted a great deal of attention \cite{WPT_survey,kaibin_zhou,bi}. 
A new energy harvesting solution called wireless power transfer is adopted in recent research 
on energy-constrained wireless networks \cite{ZhangHo,XunZhou,LiangLiu,JuOld,Ali_odd,Lee}. 
Generally speaking, the key idea is that a wireless node could capture RF signal sent by a source node and convert it into direct current to charge its battery, then use it for signal processing or transmission. 
In \cite{ZhangHo,XunZhou,LiangLiu}, the authors considered the scenario that the destination  simultaneously receives wireless information and harvests wireless power transfered by the source. 
Motivated by these works, the authors in \cite{JuOld,Ali_odd,Lee} studied 
how the wireless nodes can make use of the harvested energy from wireless power transfer to enable communications.
The wireless power transfer process can be fully
controlled, and hence, can be used in wireless networks with critical quality-of-service constrained applications, such as secure wireless communications.
%
%
In \cite{HongXing_secure,LiangLiu_secure}, the authors considered secure communications with one information receiver and one (or several) wireless energy-harvesting eavesdropper(s). 
In \cite{china_TVT}, the authors studied the coexistence of 
three destination types in a network: an information receiver, a receiver for harvesting wireless energy and an eavesdropper.
%
%
In \cite{Ng_secure}, the authors considered the wireless communication network with eavesdroppers and two types of legal receivers which can receive information and harvest wireless energy at the same time: desired receiver and idle receiver, while the idle receivers are treated as potential eavesdroppers.
All these works on secure communication did not explicitly study how the harvested energy at the receiver is used. 

\subsection{Our Work and Contribution}
This paper considers a scenario that the network designer wants to establish secure communication between a pair of source-destination devices with minimal cost.
		To this end, a simple passive device is deployed nearby as a helper.
		Such a device does not have connection to power line and is only activated during secure communication.
		The requirements of simplicity and low cost bring important challenges:
		the helping device should have low complexity in its design and operation, with a low-cost energy harvesting method to enable its operation when needed. 
		Consequently, the helping device should ideally have very little workload of online computation and minimal coordination or information exchange with the source-destination pair.		
\par	
		To solve the above-mentioned secure communication design problem, we propose to use a wireless-powered friendly jammer as the helping device, where the jammer harvests energy via wireless power transfer from the source node.
		The energy harvesting circuit (consisting of diode(s) and a passive low-pass filter~\cite{WPT_survey,XunZhou}) is very simple and cost effective.
		More importantly, such a design allows us to control the energy harvesting process for the jammer, which is very different from the conventional energy harvesting methods that rely on uncontrollable energy sources external to the communication network.
		We use a simple time-switching protocol~\cite{ZhangHo,JuOld,Ali}{, where power transfer (PT) and information transmission (IT) are separated in time.
		In this regard, the time allocation between PT and IT must be carefully designed in order to achieve the best possible throughput performance.
		We solve this problem by optimizing the jamming power, which indirectly gives the best time allocation for achieving the maximum throughput while satisfying a given secrecy constraint.
		We further optimize the rate parameters of secure communication.
		All design parameters are optimized offline with only statistical knowledge of the wireless channels.		
	}

The main contributions of this work are summarized below:
\begin{itemize}
	\item 
	The novelty of the work lies in the design of a communication protocol that provides secure communication using an energy-constrained jamming node wirelessly powered by the source node.	
	The protocol sets a target jamming power and switches between IT and PT depending on whether the available energy at the jammer meets the target power or not.
	\item 
	We study the long-term behavior of the proposed communication protocol and 
	derive a closed-form expression of the probability of IT. 
	Based on this, we obtain the achievable throughput of the protocol with fixed-rate transmission.
	\item 
	We optimize the rate parameters to achieve the maximum throughput while satisfying a constraint on the secrecy outage probability. 
	Further design insights are obtained by considering the high SNR regime and the large number of antennas regime. 
	We show that when  the jammer has a single antenna, increasing the source transmit power 
	quickly makes the throughput converge to an upper bound. 
	However, when the jammer has multiple antennas, increasing the source transmit power or the number of jammer antennas improves the throughput significantly.
\end{itemize}

Our work is different from the following most related studies: 
	In~\cite{XiZhang}, the authors considered a MISO secure communication scenario, without the help of an individual jammer.
		Different from~{\cite{XiZhang}}, we consider using wireless-powered jammer to help the secure communication.
		Therefore, in our analysis, we study the cooperation of jammer and source and design the protocol to balance the time spent on PT and IT, in order to achieve the maximum throughput of the secure communication.
	In~\cite{Ali}, the authors considered using a wireless-powered relay to help the point-to-point communication.
		Different from~\cite{Ali}, we consider a secure communication scenario.
	In our analysis, we optimize the jamming power and rate parameters for secure communication, which was not considered in \cite{Ali}.	
	In~\cite{Muk}, the authors designed jamming signal of energy harvesting jammer to help the secure communication based on the knowledge of the uncontrollable energy harvesting process.
		Different from~\cite{Muk}, we consider using wireless-powered jammer where the wireless power transfer process is totally controllable. In our work, we jointly design the wireless power transfer process and the communication process. Therefore, the design approach is fundamentally different between~\cite{Muk} and our work.

\par
The remainder of this paper is organized as follows.
Section II presents the system model.
Section III proposes the secure communication protocol.
Section IV analyzes the protocol and derives the achievable throughput.
Section V formulates an optimization problem for secrecy performance, and gives the optimal design.
Section VI presents numerical results.
Finally, conclusions are given in Section VII.\par

\section{System Model}
We consider a communication scenario where a source node ($S$) communicates with a destination node ($D$) in the presence of a passive eavesdropper ($E$) with the help of a friendly jammer ($J$), as illustrated in Fig. \ref{fig:Sys}. 
We assume that the jammer has $N_J$ antennas ($N_J \geq 1$), while all the other nodes are equipped with a single antenna only.
Also we assume that the eavesdropper is just another communication node in the same network which should not have access to the information transmitted from the source to the destination. Therefore, the locations of all nodes are public knowledge.

\subsection{Jammer Model}
In this work, the jammer is assumed to be an energy constrained node with no power of its own and having a rechargeable battery with infinite capacity \cite{yang_packet,LiangLiu,Ali}. In order to make use of the jammer, the source node wirelessly charges the jammer via wireless power transfer. Once the jammer harvests sufficient energy, it can be used for transmitting friendly jamming signals to enhance the security of the communication between the source and the destination. We assume that the jammer's energy consumption is dominated by the jamming signal transmission, while the other energy consumption, e.g.,  due to the signal processing, is relatively insignificant and hence ignored for simplicity \cite{Ali_odd,XunZhou}.

\begin{figure}[t]
	\renewcommand{\captionlabeldelim}{ }
	\renewcommand{\captionfont}{\small} \renewcommand{\captionlabelfont}{\small}
	\centering
	\usetikzlibrary{arrows}
	\usetikzlibrary{arrows}
	\begin{tikzpicture}[scale = 0.45]

\draw [ultra thick] (-8.4,3.35) ellipse (0.5 and 0.5) node{$S$};
\draw [ultra thick] (-6.95,0.8) ellipse (0.5 and 0.5)node{$J$};
\draw [ultra thick] (-3.1,3.7) ellipse (0.5 and 0.5)node{$D$};
\draw [ultra thick,fill] (-3.9,1.25) ellipse (0.5 and 0.5);
\node [white] at (-3.9,1.25) {$\bm E$};

\draw (-7.95,1.95) -- (-7.55,1.95) -- (-7.75,1.8) node (v1) {} -- (-7.95,1.95);
\draw (-7.2,1.95) -- (-6.8,1.95) -- (-7,1.8) node (v3) {} -- (-7.2,1.95);

\draw (-7.75,1.55) node (v2) {} -- (-7,1.55) node (v4) {};

\draw (-7.75,1.8) -- (-7.75,1.55);
\draw (-7,1.8) -- (-7,1.55);
\node at (-7.38,1.75) {...};
\draw (-7.4,1.55) -- (-7.4,1.1);

\draw [-triangle 60,  ultra thick](-8.05,2.85) -- (-7.55,2.1);

\draw [ultra thick](-9.4,0) -- (10,0);
\draw [ultra thick](-0.5,0) -- (-0.5,4.35);

\draw [ultra thick] (2.05,3.35) ellipse (0.5 and 0.5) node{$S$};
\draw [ultra thick] (3.5,0.8) ellipse (0.5 and 0.5)node{$J$};
\draw [ultra thick] (7.35,3.7) ellipse (0.5 and 0.5)node{$D$};
\draw [ultra thick,fill] (6.55,1.25) ellipse (0.5 and 0.5);
\node [white] at (6.55,1.25) {$\bm E$};
\draw (2.6,2) -- (3,2) -- (2.8,1.85) node (v1) {} -- (2.6,2);
\draw (3.35,2) -- (3.75,2) -- (3.55,1.85) node (v3) {} -- (3.35,2);

\draw (2.8,1.6) node (v2) {} -- (3.55,1.6) node (v4) {};

\draw (2.8,1.85) -- (2.8,1.6);
\draw (3.55,1.85) -- (3.55,1.6);
\node at (3.17,1.8) {...};
\draw (3.15,1.6) -- (3.15,1.25);

\draw [-stealth,  ultra thick, dotted](2.65,3.3) node (v5) {} -- (6.05,1.5);

\draw [-stealth,  ultra thick](2.7,3.3) -- (6.75,3.7);

\draw [loosely dashed,  -latex, ultra thick](3.9,1.9) node (v6) {} -- (6.6,3.45);
\draw [loosely dashed,  -latex, ultra thick](v6) -- (5.8,1.4);

\node [font=\small] at (-5.3,-0.6) {Power Transfer (PT) Block};
\node [font=\small] at (4.6,-0.6) {Information Transmission (IT) Block};

\draw [ultra thick] (-8.4,-1.85) ellipse (0.5 and 0.5)node{$S$};
\draw [ultra thick] (-8.4,-3.05) ellipse (0.5 and 0.5)node{$D$};
\draw [ultra thick] (-8.4,-4.25) ellipse (0.5 and 0.5)node{$J$};
\draw [ultra thick,fill] (-8.4,-5.45) ellipse (0.5 and 0.5);
\node [white] at (-8.4,-5.45) {$\bm E$};
\node [right,font=\small] at (-7.8,-1.85) {Source node};
\node [right,font=\small] at (-7.8,-3.05) {Destination node};
\node [right,font=\small] at (-7.8,-4.25) {Jamming node};
\node [right,font=\small] at (-7.8,-5.45) {Eavesdropper node};

\draw [ultra thick, -triangle 60](-0.55,-1.85) -- (0.75,-1.85);
\draw [ultra thick,  -stealth](-0.55,-3.15) -- (0.75,-3.15);
\draw [loosely dashed,  -latex, ultra thick](-0.55,-4.3) -- (0.75,-4.3);
\draw [-stealth,  ultra thick, dotted](-0.55,-5.5) -- (0.75,-5.5);

\node [right,font=\small] at (0.85,-1.85) {Power transfer link};
\node [right,font=\small] at (0.85,-3.15) {Information transmission link};
\node [right,font=\small] at (0.85,-4.3) {Jamming channels};
\node [right,font=\small] at (0.85,-5.5) {Eavesdropper link};

	\end{tikzpicture}
	\vspace*{-0.5cm}
	\caption{System model with illustration of the power transfer and information transmission phases.}
	\label{fig:Sys}
	\vspace*{-0.4cm}
\end{figure}
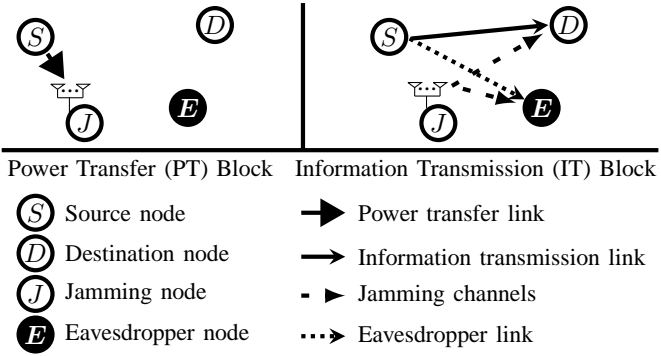

\subsection{Channel Assumptions}
We assume that all the channel links are composed of large-scale path loss with exponent $m$ and statistically independent small-scale Rayleigh fading. 
We denote the inter-node distance of links $S\rightarrow J$, $S\rightarrow D$, $J\rightarrow D$, $S\rightarrow E$ and $J\rightarrow E$ by $d_{SJ}$, $d_{SD}$, $d_{JD}$, $d_{SE}$ and $d_{JE}$, respectively. 
The fading channel gains of the links $S\rightarrow J$, $S\rightarrow D$, $S\rightarrow E$, $J\rightarrow E$ and $J\rightarrow D$ are denoted by $\bm h_{SJ}$, $h_{SD}$, $h_{SE}$, $\bm h_{JE}$, $\bm h_{JD}$, respectively. 
These fading channel gains are modeled as quasi-static frequency non-selective parameters, which means that they are constant over the block time of $T$ seconds and independent and identically distributed between blocks. 
Consequently, each element of these complex fading channel coefficients are circular symmetric complex Gaussian random variables with zero mean and unit variance.
In this paper, 
we make the following assumptions on channel state information (CSI) and noise power:
\begin{itemize}
	\item The CSI ($h_{SD}$ and $\bm h_{JD}$) is assumed to be perfectly available at both the transmitter and receiver sides. This allows benchmark system performance to be determined.	
	\item The CSI of the eavesdropper is only known to itself.
	\item Noise power at the eavesdropper is zero in line with \cite{zhou_trans}, which corresponds to the worst case scenario.
\end{itemize}
\subsection{Transmission Phases}
The secure communication with wireless-powered jammer takes places in two phases: (i) power transfer (PT) phase and (ii) information transmission (IT) phase, as shown in Fig. \ref{fig:Sys}.
During the PT phase, the source transfers power to the jammer by sending a radio signal with power $\mathcal{P}_s$. The jammer receives the radio signal, converts it to a direct current signal and stores the energy in its battery. 
During the IT phase, the jammer sends jamming signal to the eavesdropper with power $\mathcal{P}_J$ by using the stored energy in the battery.
At the same time,
the source transmits the information signal to the destination with power $\mathcal{P}_s$ under the protection of the jamming signal.
%
We define the \textit{information transmission probability} as the probability of 
the communication process being in the IT phase and denote it by~$p_{tx}$.

\subsection{Secure Encoding and Performance Metrics}
We consider confidential transmission 
between the source and the destination, 
using Wyner's wiretap code \cite{wyner}. 
Specifically, there are two rate parameters of the wiretap code, namely the rate of codeword transmission, denoted by $R_t$, and the rate of secret information, denoted by $R_s$. 
%
The positive rate difference $R_t - R_s$ is the cost to provide secrecy against the eavesdropper. 
A $M$-length wiretap code is constructed by generating $2^{M R_t}$ codewords $x^{M}\left(w,v\right) $ of the length $M$, where $w = 1,2,...,2^{M R_s}$ and $v = 1,2,...,2^{M \left(R_t-R_s\right)}$.
For each message index $w$, the value of $v$ is selected randomly with uniform probability from $\left\lbrace 1,2,...,2^{M\left(R_t-R_s\right)} \right\rbrace$, and the constructed codeword to be transmitted is $x^M\left(w,v\right)$. 
Clearly, reliable transmission from the source to the destination cannot be achieved when $R_t> C_d$, where $C_d$ denotes the channel capacity of $S \rightarrow D$ link. This event is defined as connection outage event. 
From \cite{wyner}, perfect secrecy cannot be achieved when $R_t - R_s < C_e$, where $C_e$ denotes the fading channel capacity of $S\rightarrow E$ link. This event is defined as secrecy outage event. 
In this work, we consider fixed rate transmission, which means $R_t$ and $R_s$ are fixed and chosen offline following \cite{Tang,Zhou_rethinking}. 

Since we consider quasi-static fading channel, we use outage based measures as considered in \cite{Tang,Zhou_rethinking}. Specifically, the connection outage probability and secrecy outage probability are defined, respectively, as
\begin{align}
\label{p_co}
p_{co}=&\mathbb{P}\left\lbrace R_t  > C_d  \right \rbrace,\\
\label{p_so_first}
p_{so}=&\mathbb{P}\left\lbrace R_t-R_s < C_e  \right \rbrace,
\end{align}
where $\mathbb{P}\left\lbrace \nu \right\rbrace$ denotes the probability for success of event $\nu$. 
Note that 
the connection outage probability is a measure of the fading channel quality of the $S \rightarrow D$ link. 
Since the current CSI is available at the legitimate nodes, the source can always suspend transmission when connection outage occurs. 
This is easy to realize by one-bit feedback from the destination. 
Therefore, in this work, connection outage leads to suspension of IT but not decoding error at the destination. 
%
\par

Our figure of merit is the throughput, $\pi$, 
which is the average number of bits of confidential information received at the destination per unit time \cite{XiZhang,Zhou_rethinking}, and is given by
\begin{equation} \label{defi_B}
\pi = p_{tx}  R_s.
\end{equation}
As we will see in Section IV, the information transmission probability $p_{tx}$ contains the connection outage probability $p_{co}$.
\par
It is important to note that a trade-off exists between throughput achieved at the destination and secrecy against the eavesdropper (measured by the secrecy outage probability). 
For example, increasing $R_s$ would increase $\pi$ in \eqref{defi_B}, but also increase $p_{so}$ in \eqref{p_so_first}. 
This trade-off will be investigated later in Section V in this paper. 
%
\section{Proposed Secure Communication Protocol}

In this section, we propose a simple fixed-power and fixed-rate secure communication protocol, which employs a wireless-powered jammer. Note that more sophisticated power and rate adaptation strategies at the source are possible but outside the scope of this paper.

\subsection{Transmission Protocol}
We consider the communication in blocks of $T$ seconds, each block being either a PT or an IT block. 
Intuitively, IT should happen when the jammer has sufficient energy for jamming and the $S\rightarrow D$ link is in a good condition to ensure successful information reception at the destination. 
We define the two conditions for a block to be used for IT as follows:
\begin{itemize}
	\item 
	
	At the beginning of the block, 
	the jammer has enough energy, $\mathcal{P}_J T$, to support jamming with power $\mathcal{P}_J$ over an information transmission block of $T$ seconds, and
	
	\item the link $S\rightarrow D$ does not suffer connection outage, which means it can support the codeword transmission rate $R_t$ from the source to the destination.
\end{itemize}

Note that both conditions are checked at the start of each block using the knowledge of the actual amount of energy in the jammer's battery and the instantaneous CSI of $S\rightarrow D$ link, and both conditions must be satisfied simultaneously for the block to be an IT block.
If the first condition is not satisfied, then the block is used for PT and we refer to it as a \textit{dedicated PT block}. 
If the first condition is satisfied while the second condition is not, then the block is still used for PT but we refer it as an \textit{opportunistic PT block}. Note that $\mathcal{P}_J$ is a design parameter in the proposed protocol.
\par

For an accurate description of the transmission process, we define a PT-IT cycle as a sequence of blocks which either consists of a single IT block or a sequence of PT blocks followed by an IT block. 
Let discrete random variables $X$ and $Y$ ($X,Y=0,1,2,...$) denote the number of dedicated and opportunistic PT blocks in a PT-IT cycle, respectively.
In our proposed protocol, the following four types of PT-IT cycles are possible:
\begin{enumerate}
	\item $X>0,\ Y=0$, i.e., PT-IT cycle contains $X$ dedicated PT blocks followed by an IT block. This is illustrated as the $k\textrm{th}$ PT-IT cycle in Fig. \ref{fig:Xyblock}.
	\item $X>0,\ Y>0$, i.e., PT-IT cycle contains $X$ dedicated PT blocks and $Y$ opportunistic PT blocks followed by an IT block. This is illustrated as the $\left(k+1\right)\textrm{th}$ PT-IT cycle in Fig. \ref{fig:Xyblock}.
	\item $X=0,\ Y>0$, i.e., PT-IT cycle contains $Y$ opportunistic PT blocks followed by an IT block. This is illustrated as the $\left(k+2\right)\textrm{th}$ PT-IT cycle in Fig. \ref{fig:Xyblock}.
	\item $X=0,\ Y=0$, i.e., PT-IT cycle contains one IT block only. This is illustrated as the $\left(k+3\right)\textrm{th}$ PT-IT cycle in Fig. \ref{fig:Xyblock}.
\end{enumerate}

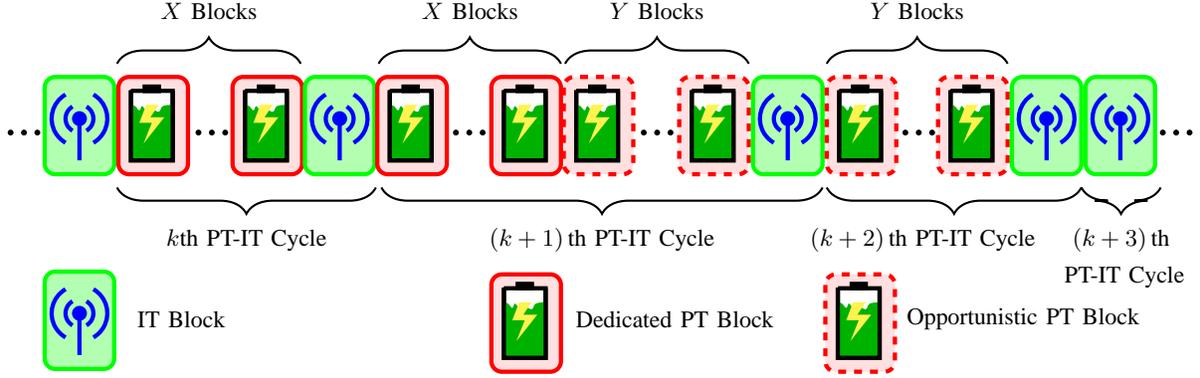
\begin{figure*}[!t]
	\renewcommand{\captionlabeldelim}{ }
	\renewcommand{\captionfont}{\small} \renewcommand{\captionlabelfont}{\small}
	\centering
	\begin{tikzpicture}[scale=0.37]
	\draw  [green,ultra thick,draw opacity = 100, fill=green!30,rounded corners,ultra thick] (-18.5,3) rectangle (-16,-0.5);
	
	\node [circle,fill=blue,scale=0.6] at (-17.25,1.5) {};
	\draw [blue,ultra thick](-17.75,2.366) arc (120.0007:240:1);
	\draw [blue,ultra thick](-16.75,2.366) arc (59.9993:-60:1);
	
	\draw  [blue,ultra thick] (-17.5,1.933) arc (120.0007:240:0.5);
	\draw  [blue,ultra thick] (-17,1.933) arc (59.9993:-60:0.5);
	
	\draw  [blue,ultra thick] (-17.25,1.5) -- (-17.25,0);

	\draw  [red,ultra thick,draw opacity = 100, fill=pink!50,rounded corners,ultra thick]  (-15.85,3) rectangle (-13.35,-0.5);
	
	\draw [fill=white,ultra thick] (-15.4,2.53)  -- (-15.4,0) node (v2) {} -- (-13.9,0) node (v8) {} -- (-13.9,2.5) -- (-15.4,2.5);
	\draw [fill = white,ultra thick]  (-14.9,2.65) rectangle (-14.4,2.5);
	%
	\draw [fill=black!30!green,draw opacity=0] plot[smooth, tension=.7] coordinates {(-15.4,2) (-15.25,1.95) (-15.1,1.8) (-15,1.95) (-14.85,1.8) (-14.65,1.85) (-14.4,1.75) (-14.25,1.95) (-13.9,2)}--(-13.9,1.7) -- (-13.9,0) -- (-15.4,0) -- (-15.4,2);
	\draw [ultra thick] (-13.9,2.5) -- (-13.9,0) -- (-15.4,0) -- (-15.4,2.5);
	\draw [white, fill=yellow!80,draw opacity=0](-14.9,2)  -- (-14.4,2) -- (-14.7,1.55) -- (-14.2,1.55) -- (-14.95,0.6) -- (-14.65,1.35) -- (-15.1,1.35) -- (-14.85,2);
	\draw [ultra thick,]  (-14.9,2.65) rectangle (-14.4,2.5);
	%
	%
	\draw  [red,ultra thick,draw opacity = 100, fill=pink!50,rounded corners,ultra thick]  (-11.75,3) rectangle (-9.25,-0.5);
	
	\draw [fill=white,ultra thick] (-11.3,2.53)  -- (-11.3,0) node (v2) {} -- (-9.8,0) node (v8) {} -- (-9.8,2.5) -- (-11.3,2.5);
	\draw [fill = white,ultra thick]  (-10.8,2.65) rectangle (-10.3,2.5);
	%
	\draw [fill=black!30!green,draw opacity=0] plot[smooth, tension=.7] coordinates {(-11.3,2) (-11.15,1.95) (-11,1.8) (-10.9,1.95) (-10.75,1.8) (-10.55,1.85) (-10.3,1.75) (-10.15,1.95) (-9.8,2)}--(-9.8,1.7) -- (-9.8,0) -- (-11.3,0) -- (-11.3,2);
	\draw [ultra thick] (-9.8,2.5) -- (-9.8,0) -- (-11.3,0) -- (-11.3,2.5);
	\draw [white, fill=yellow!80,draw opacity=0](-10.8,2)  -- (-10.3,2) -- (-10.6,1.55) -- (-10.1,1.55) -- (-10.85,0.6) -- (-10.55,1.35) -- (-11,1.35) -- (-10.75,2);
	\draw [ultra thick,]  (-10.8,2.65) rectangle (-10.3,2.5);

	\draw  [green,ultra thick,draw opacity = 100, fill=green!30,rounded corners,ultra thick] (-9.15,3) rectangle (-6.65,-0.5);
	
	\node [circle,fill=blue,scale=0.6] at (-7.9,1.5) {};
	\draw [blue,ultra thick](-8.4,2.366) arc (120.0007:240:1);
	\draw [blue,ultra thick](-7.4,2.366) arc (59.9993:-60:1);
	
	\draw  [blue,ultra thick] (-8.15,1.933) arc (120.0007:240:0.5);
	\draw  [blue,ultra thick] (-7.65,1.933) arc (59.9993:-60:0.5);
	
	\draw  [blue,ultra thick] (-7.9,1.5) -- (-7.9,0);
	\draw  [red,ultra thick,draw opacity = 100, fill=pink!50,rounded corners,ultra thick]  (-6.55,3) rectangle (-4.05,-0.5);
	
	\draw [fill=white,ultra thick] (-6.1,2.53)  -- (-6.1,0) node (v2) {} -- (-4.6,0) node (v8) {} -- (-4.6,2.5) -- (-6.1,2.5);
	\draw [fill = white,ultra thick]  (-5.6,2.65) rectangle (-5.1,2.5);
	%
	\draw [fill=black!30!green,draw opacity=0] plot[smooth, tension=.7] coordinates {(-6.1,2) (-5.95,1.95) (-5.8,1.8) (-5.7,1.95) (-5.55,1.8) (-5.35,1.85) (-5.1,1.75) (-4.95,1.95) (-4.6,2)}--(-4.6,1.7) -- (-4.6,0) -- (-6.1,0) -- (-6.1,2);
	\draw [ultra thick] (-4.6,2.5) -- (-4.6,0) -- (-6.1,0) -- (-6.1,2.5);
	\draw [white, fill=yellow!80,draw opacity=0](-5.6,2)  -- (-5.1,2) -- (-5.4,1.55) -- (-4.9,1.55) -- (-5.65,0.6) -- (-5.35,1.35) -- (-5.8,1.35) -- (-5.55,2);
	\draw [ultra thick,]  (-5.6,2.65) rectangle (-5.1,2.5);
	
	\draw  [red,ultra thick,draw opacity = 100, fill=pink!50,rounded corners,ultra thick]  (-2.45,3) rectangle (0.05,-0.5);
	
	\draw [fill=white,ultra thick] (-2,2.53)  -- (-2,0) node (v2) {} -- (-0.5,0) node (v8) {} -- (-0.5,2.5) -- (-2,2.5);
	\draw [fill = white,ultra thick]  (-1.5,2.65) rectangle (-1,2.5);
	%
	\draw [fill=black!30!green,draw opacity=0] plot[smooth, tension=.7] coordinates {(-2,2) (-1.85,1.95) (-1.7,1.8) (-1.6,1.95) (-1.45,1.8) (-1.25,1.85) (-1,1.75) (-0.85,1.95) (-0.5,2)}--(-0.5,1.7) -- (-0.5,0) -- (-2,0) -- (-2,2);
	\draw [ultra thick] (-0.5,2.5) -- (-0.5,0) -- (-2,0) -- (-2,2.5);
	\draw [white, fill=yellow!80,draw opacity=0](-1.5,2)  -- (-1,2) -- (-1.3,1.55) -- (-0.8,1.55) -- (-1.55,0.6) -- (-1.25,1.35) -- (-1.7,1.35) -- (-1.45,2);
	\draw [ultra thick,]  (-1.5,2.65) rectangle (-1,2.5);
	
	\draw  [dashed,red,ultra thick,draw opacity = 100, fill=pink!50,rounded corners,ultra thick]  (0.15,3) rectangle (2.65,-0.5);
	
	\draw [fill=white,ultra thick] (0.6,2.53)  -- (0.6,0) node (v2) {} -- (2.1,0) node (v8) {} -- (2.1,2.5) -- (0.6,2.5);
	\draw [fill = white,ultra thick]  (1.1,2.65) rectangle (1.6,2.5);
	%
	\draw [fill=black!30!green,draw opacity=0] plot[smooth, tension=.7] coordinates {(0.6,2) (0.75,1.95) (0.9,1.8) (1,1.95) (1.15,1.8) (1.35,1.85) (1.6,1.75) (1.75,1.95) (2.1,2)}--(2.1,1.7) -- (2.1,0) -- (0.6,0) -- (0.6,2);
	\draw [ultra thick] (2.1,2.5) -- (2.1,0) -- (0.6,0) -- (0.6,2.5);
	\draw [white, fill=yellow!80,draw opacity=0](1.1,2)  -- (1.6,2) -- (1.3,1.55) -- (1.8,1.55) -- (1.05,0.6) -- (1.35,1.35) -- (0.9,1.35) -- (1.15,2);
	\draw [ultra thick,]  (1.1,2.65) rectangle (1.6,2.5);
	
	%
	\draw  [dashed,red,ultra thick,draw opacity = 100, fill=pink!50,rounded corners,ultra thick]  (4.3,3) rectangle (6.8,-0.5);
	
	\draw [fill=white,ultra thick] (4.75,2.53)  -- (4.75,0) node (v2) {} -- (6.25,0) node (v8) {} -- (6.25,2.5) -- (4.75,2.5);
	\draw [fill = white,ultra thick]  (5.25,2.65) rectangle (5.75,2.5);
	%
	\draw [fill=black!30!green,draw opacity=0] plot[smooth, tension=.7] coordinates {(4.75,2) (4.9,1.95) (5.05,1.8) (5.15,1.95) (5.3,1.8) (5.5,1.85) (5.75,1.75) (5.9,1.95) (6.25,2)}--(6.25,1.7) -- (6.25,0) -- (4.75,0) -- (4.75,2);
	\draw [ultra thick] (6.25,2.5) -- (6.25,0) -- (4.75,0) -- (4.75,2.5);
	\draw [white, fill=yellow!80,draw opacity=0](5.25,2)  -- (5.75,2) -- (5.45,1.55) -- (5.95,1.55) -- (5.2,0.6) -- (5.5,1.35) -- (5.05,1.35) -- (5.3,2);
	\draw [ultra thick,]  (5.25,2.65) rectangle (5.75,2.5);

	\draw  [green,ultra thick,draw opacity = 100, fill=green!30,rounded corners,ultra thick] (6.95,3) rectangle (9.45,-0.5);
	
	\node [circle,fill=blue,scale=0.6] at (8.2,1.5) {};
	\draw [blue,ultra thick](7.7,2.366) arc (120.0007:240:1);
	\draw [blue,ultra thick](8.7,2.366) arc (59.9993:-60:1);
	
	\draw  [blue,ultra thick] (7.95,1.933) arc (120.0007:240:0.5);
	\draw  [blue,ultra thick] (8.45,1.933) arc (59.9993:-60:0.5);
	
	\draw  [blue,ultra thick] (8.2,1.5) -- (8.2,0);
	
	\draw  [dashed,red,ultra thick,draw opacity = 100, fill=pink!50,rounded corners,ultra thick]  (9.6,3) rectangle (12.1,-0.5);
	
	\draw [fill=white,ultra thick] (10.05,2.53)  -- (10.05,0) node (v2) {} -- (11.55,0) node (v8) {} -- (11.55,2.5) -- (10.05,2.5);
	\draw [fill = white,ultra thick]  (10.55,2.65) rectangle (11.05,2.5);
	%
	\draw [fill=black!30!green,draw opacity=0] plot[smooth, tension=.7] coordinates {(10.05,2) (10.2,1.95) (10.35,1.8) (10.45,1.95) (10.6,1.8) (10.8,1.85) (11.05,1.75) (11.2,1.95) (11.55,2)}--(11.55,1.7) -- (11.55,0) -- (10.05,0) -- (10.05,2);
	\draw [ultra thick] (11.55,2.5) -- (11.55,0) -- (10.05,0) -- (10.05,2.5);
	\draw [white, fill=yellow!80,draw opacity=0](10.55,2)  -- (11.05,2) -- (10.75,1.55) -- (11.25,1.55) -- (10.5,0.6) -- (10.8,1.35) -- (10.35,1.35) -- (10.6,2);
	\draw [ultra thick,]  (10.55,2.65) rectangle (11.05,2.5);
	
	\draw  [dashed,red,ultra thick,draw opacity = 100, fill=pink!50,rounded corners,ultra thick]  (13.6,3) rectangle (16.1,-0.5);
	
	\draw [fill=white,ultra thick] (14.05,2.53)  -- (14.05,0) node (v2) {} -- (15.55,0) node (v8) {} -- (15.55,2.5) -- (14.05,2.5);
	\draw [fill = white,ultra thick]  (14.55,2.65) rectangle (15.05,2.5);
	%
	\draw [fill=black!30!green,draw opacity=0] plot[smooth, tension=.7] coordinates {(14.05,2) (14.2,1.95) (14.35,1.8) (14.45,1.95) (14.6,1.8) (14.8,1.85) (15.05,1.75) (15.2,1.95) (15.55,2)}--(15.55,1.7) -- (15.55,0) -- (14.05,0) -- (14.05,2);
	\draw [ultra thick] (15.55,2.5) -- (15.55,0) -- (14.05,0) -- (14.05,2.5);
	\draw [white, fill=yellow!80,draw opacity=0](14.55,2)  -- (15.05,2) -- (14.75,1.55) -- (15.25,1.55) -- (14.5,0.6) -- (14.8,1.35) -- (14.35,1.35) -- (14.6,2);
	\draw [ultra thick,]  (14.55,2.65) rectangle (15.05,2.5);
	
	\draw  [green,ultra thick,draw opacity = 100, fill=green!30,rounded corners,ultra thick] (16.25,3) rectangle (18.75,-0.5);
	
	\node [circle,fill=blue,scale=0.6] at (17.5,1.5) {};
	\draw [blue,ultra thick](17,2.366) arc (120.0007:240:1);
	\draw [blue,ultra thick](18,2.366) arc (59.9993:-60:1);
	
	\draw  [blue,ultra thick] (17.25,1.933) arc (120.0007:240:0.5);
	\draw  [blue,ultra thick] (17.75,1.933) arc (59.9993:-60:0.5);
	
	\draw  [blue,ultra thick] (17.5,1.5) -- (17.5,0);
	
	\draw  [green,ultra thick,draw opacity = 100, fill=green!30,rounded corners,ultra thick] (18.9,3) rectangle (21.4,-0.5);
	
	\node [circle,fill=blue,scale=0.6] at (20.15,1.5) {};
	\draw [blue,ultra thick](19.65,2.366) arc (120.0007:240:1);
	\draw [blue,ultra thick](20.65,2.366) arc (59.9993:-60:1);
	
	\draw  [blue,ultra thick] (19.9,1.933) arc (120.0007:240:0.5);
	\draw  [blue,ultra thick] (20.4,1.933) arc (59.9993:-60:0.5);
	
	\draw  [blue,ultra thick] (20.15,1.5) -- (20.15,0);
	
	\node [font=\huge] at (-19.25,1) {...};
	\node [font=\huge]at (-12.5,1) {...};
	\node [font=\huge]at (-3.25,1) {...};
	\node [font=\huge]at (3.5,1) {...};
	\node [font=\huge] at (12.9,1) {...};
	\node [font=\huge] at (22.15,1) {...};

	\draw [thick, decoration={brace,  mirror, raise=5,amplitude=10},decorate] (-15.9,-0.5) -- (-6.6,-0.5) node [midway,yshift=-25,font=\small] {$k\textrm{th}$ PT-IT Cycle};
	\draw [thick, decoration={brace,  mirror, raise=5,amplitude=10}, decorate] (-6.4,-0.5) -- (9.45,-0.5) node [midway,yshift=-25,font=\small] {$\left(k+1\right)\textrm{th}$ PT-IT Cycle};
	\draw [thick, decoration={brace,  mirror, raise=5,amplitude=10}, decorate] (9.6,-0.5) -- (18.75,-0.5) node [xshift=-60,yshift=-25,font=\small] {$\left(k+2\right)\textrm{th}$ PT-IT Cycle};
	\draw [thick, decoration={brace,  mirror, raise=5,amplitude=10}, decorate] (18.75,-0.5) -- (21.6,-0.5) node [yshift=-25,midway,font=\small] {$\left(k+3\right)\textrm{th}$};
	\node [font=\small]at (20.3,-4.2) {PT-IT Cycle};

	\draw [thick, decoration={brace,   raise=5,amplitude=10}, decorate] (-15.9,3) -- (-9.3,3) node [black,midway,yshift=25,font=\small] {$X$ Blocks};
	\draw [thick, decoration={brace,   raise=5,amplitude=10}, decorate] (-6.5,3) -- (0.05,3) node [black,midway,yshift=25,font=\small] {$X$ Blocks};
	\draw [thick, decoration={brace,   raise=5,amplitude=10}, decorate] (0.2,3) -- (6.8,3) node [black,midway,yshift=25,font=\small] {$Y$ Blocks};
	\draw [thick, decoration={brace,   raise=5,amplitude=10}, decorate] (9.6,3) -- (16.1,3) node [black,midway,yshift=25,font=\small] {$Y$ Blocks};
	\draw  [green,ultra thick,draw opacity = 100, fill=green!30,rounded corners,ultra thick] (-18.5,-4) rectangle (-16,-7.5);
	
	\node [circle,fill=blue,scale=0.6] at (-17.25,-5.5) {};
	\draw [blue,ultra thick](-17.75,-4.634) arc (120.0007:240:1);
	\draw [blue,ultra thick](-16.75,-4.634) arc (59.9993:-60:1);
	
	\draw  [blue,ultra thick] (-17.5,-5.067) arc (120.0007:240:0.5);
	\draw  [blue,ultra thick] (-17,-5.067) arc (59.9993:-60:0.5);
	
	\draw  [blue,ultra thick] (-17.25,-5.5) -- (-17.25,-7);
	
	%
	%
	%
	

	\draw  [red,ultra thick,draw opacity = 100, fill=pink!50,rounded corners,ultra thick]  (-2.45,-4.1) rectangle (0.05,-7.6);
	
	\draw [fill=white,ultra thick] (-2,-4.57)  -- (-2,-7.1) node (v2) {} -- (-0.5,-7.1) node (v8) {} -- (-0.5,-4.6) -- (-2,-4.6);
	\draw [fill = white,ultra thick]  (-1.5,-4.45) rectangle (-1,-4.6);
	%
	\draw [fill=black!30!green,draw opacity=0] plot[smooth, tension=.7] coordinates {(-2,-5.1) (-1.85,-5.15) (-1.7,-5.3) (-1.6,-5.15) (-1.45,-5.3) (-1.25,-5.25) (-1,-5.35) (-0.85,-5.15) (-0.5,-5.1)}--(-0.5,-5.4) -- (-0.5,-7.1) -- (-2,-7.1) -- (-2,-5.1);
	\draw [ultra thick] (-0.5,-4.6) -- (-0.5,-7.1) -- (-2,-7.1) -- (-2,-4.6);
	\draw [white, fill=yellow!80,draw opacity=0](-1.5,-5.1)  -- (-1,-5.1) -- (-1.3,-5.55) -- (-0.8,-5.55) -- (-1.55,-6.5) -- (-1.25,-5.75) -- (-1.7,-5.75) -- (-1.45,-5.1);
	\draw [ultra thick,]  (-1.5,-4.45) rectangle (-1,-4.6);

	\draw  [dashed,red,ultra thick,draw opacity = 100, fill=pink!50,rounded corners,ultra thick]  (9.55,-4.1) rectangle (12.05,-7.6);
	
	\draw [fill=white,ultra thick] (10,-4.57)  -- (10,-7.1) node (v2) {} -- (11.5,-7.1) node (v8) {} -- (11.5,-4.6) -- (10,-4.6);
	\draw [fill = white,ultra thick]  (10.5,-4.45) rectangle (11,-4.6);
	%
	\draw [fill=black!30!green,draw opacity=0] plot[smooth, tension=.7] coordinates {(10,-5.1) (10.15,-5.15) (10.3,-5.3) (10.4,-5.15) (10.55,-5.3) (10.75,-5.25) (11,-5.35) (11.15,-5.15) (11.5,-5.1)}--(11.5,-5.4) -- (11.5,-7.1) -- (10,-7.1) -- (10,-5.1);
	\draw [ultra thick] (11.5,-4.6) -- (11.5,-7.1) -- (10,-7.1) -- (10,-4.6);
	\draw [white, fill=yellow!80,draw opacity=0](10.5,-5.1)  -- (11,-5.1) -- (10.7,-5.55) -- (11.2,-5.55) -- (10.45,-6.5) -- (10.75,-5.75) -- (10.3,-5.75) -- (10.55,-5.1);
	\draw [ultra thick,]  (10.5,-4.45) rectangle (11,-4.6);

	\node [right,font=\small] at (-15.5,-5.7) {IT Block};
	\node [right,font=\small] at (0.25,-5.7) {Dedicated PT Block};
	\node [right,font=\small] at (12.15,-5.7) {Opportunistic PT Block};
	
	\end{tikzpicture}

	\vspace*{-0.2cm}
	\caption{Illustration of four types of PT-IT cycles.}
	\label{fig:Xyblock}
	\vspace*{-0.4cm}
\end{figure*}

\vspace*{-0.2cm}
\subsection{Long-Term Behavior} 
We are interested in the long-term behavior
%
 (rather than that during the transition stage) of the communication process determined by our proposed protocol. 
After a sufficiently long time, the behavior of the communication process falls in one of the following two cases:
%
%
%
\begin{itemize}
	\item \textit{Energy Accumulation}: 
	In this case, on average, the energy harvested at the jammer during opportunistic PT blocks is higher than the energy required during an IT block. 
	Thus, after a long time has passed, 
	the energy steadily accumulates at the jammer and 
	there is no need for dedicated PT blocks (the harvested energy by opportunistic PT blocks fully meet the energy consumption requirement at the jammer). 
	Consequently, only PT-IT cycles with $X=0$ can occur.			
	\item \textit{Energy Balanced}:
	In this case, on average, the energy harvested at the jammer during opportunistic PT blocks is lower than the energy required during an IT block. Thus, after a long time has passed, dedicated PT blocks are sometimes required to make sure that the energy harvested from both dedicated and opportunistic PT blocks equals the energy required for jamming in IT blocks on average. 
	Consequently, all four types of PT-IT cycles can occur.		
\end{itemize}
\textit{Remarks:}
		Although we have assumed infinite battery capacity for simplicity in the analysis, it is important to discuss the effect on finite battery capacity.
		In fact, our analytical result is valid for finite battery capacity as long as the battery capacity in much higher than the required jamming energy 
		$\mathcal{P}_J T$.\footnote{{{From \cite{EH_Survey},}}
		{for typical energy storage, including super-capacitor or chemical battery, the capacity easily reaches several Joules, or even several thousand Joules. While in our work, from the simulation results to be presented later, the optimal value of required jamming energy is only several micro Joules. 
			Therefore, it is reasonable to say that the battery capacity in practice is much larger than the required jamming energy.}} To be specific:
\par	{	 			 
		i) When the communication process is in the energy accumulation case, the harvested energy steadily accumulates at the jammer, thus, the energy level will always reach the maximum battery capacity after a sufficient long time and stay near the maximum capacity for the remaining time period.
		This means that the energy level in the battery is always much larger than the required jamming energy level. Thus, having a finite battery capacity has hardly any effect on the communication process, as compared with infinite capacity.
\par
		ii) When the communication process is in the energy balanced case,
		on average, the harvested energy equals the required (consumed) jamming energy.
		Therefore, the energy level mostly stays between zero and the required jamming energy level, $\mathcal{P}_J T$. This also means that the energy level in the battery can hardly approach the maximum battery capacity. Thus, having a finite battery capacity has almost no effect on the communication process, compared with infinite capacity.
\\		
{\hspace*{10pt}}Therefore, although our analysis is based on the assumption of infinite battery capacity, the analytical results still hold with practical finite battery capacity.
		}

\par
In the next section, the mathematical model for the proposed protocol is presented. 
The boundary condition between the energy accumulation and energy balanced cases is derived. 
In Section VI, we will also verify the long-term behavior through simulations. 
\vspace{-0.2cm}
\section{Protocol Analysis}
In this section, 
we analyze the proposed secure communication protocol and derive the achievable throughput for the proposed secure communication protocol. 
\vspace{-0.2cm}
\subsection{Signal Model}
{In a PT block, the source sends radio signal $x_{SJ}$ with power $\mathcal{P}_s$.
	Thus, received signal at the jammer, $\bm y_J$ is given by}
\vspace*{-6pt}
\begin{equation} \label{y_J}
{\bm y_J=\frac{1}{\sqrt{d^m_{SJ}}} \sqrt{\mathcal{P}_s} \bm h_{SJ} x_{SJ} + \bm n_J,} 
\end{equation}
where $x_{SJ}$ is the normalized signal from the source in an PT block, and $\bm n_J$ is the additive white Gaussian noise (AWGN) at the jammer. 
	{From equation (4), by ignoring the noise power,}
	{the harvested energy is given by [22]}
	\begin{equation*} 
	{\rho_J( \bm h_{SJ}) = \eta \left\vert \frac{1}{\sqrt{d^m_{SJ}}} \sqrt{\mathcal{P}_s} \bm h_{SJ} \right\vert^2 T,}
	\end{equation*}
	{where $\eta$ is the energy conversion efficiency of RF-DC conversion operation for energy storage at the jammer.
		Because the elements of $\bm h_{SJ}$ are independent identically distributed complex Gaussian random variable with normalized variance,
		we have}
	${\myexpect{\vert \bm h_{SJ} \vert^2} = N_J}$.
	{Therefore, the average harvested energy $\rho_J$ is given by}
	\begin{equation} \label{my_rho}
	{\rho_J = \myexpect{\rho_J( \bm h_{SJ}) } =
		\myexpect{\eta \frac{1}{{d^m_{SJ}}} {\mathcal{P}_s} \left\vert \bm h_{SJ} \right\vert^2 T}
		= \frac{\eta N_J \mathcal{P}_s T}{d^m_{SJ}}.} 
	\end{equation}	
\par
During an IT block,
the source transmits information-carrying signal with the protection from the jammer.
The jammer applies different signaling methods depending on its number of antennas.  
When $N_J = 1$, the jammer sends a noise-like signal $x_{JD}$ with power $\mathcal{P}_J$, affecting both the destination and the eavesdropper.
When $N_J > 1$,
by using the artificial interference generation method in \cite{zhou_trans}, the jammer generates an $N_J \times \left(N_J-1\right)$ matrix $\bm W$ which is an orthonormal basis of the null space of $\bm h_{JD}$, and also an vector $\bm v$ with $N_J -1 $ independent identically distributed complex Gaussian random elements with normalized variance.\footnote{With the assumption of zero additive noise at the eavesdropper, the null-space artificial jamming scheme works when the number of jamming antennas in larger than the number of eavesdropper antennas, as discussed in $[36]$. This condition is satisfied in this work when $N_J >1$.} Then the jammer sends $\bm {Wv}$ as jamming signal.
Thus, the received signal at the destination, $y_D$, is given by
\begin{equation} \label{yd}
y_D\!=\!
\left\lbrace
\begin{aligned}
&\frac{\sqrt{\mathcal{P}_s} }{\sqrt{d^m_{SD}}}h_{SD} x_{SD}
+ \frac{\sqrt{\mathcal{P}_J}}{\sqrt{d^m_{JD}}}  h_{JD} x_{JD}
+n_d, &N_J = 1, \\
&\frac{\sqrt{\mathcal{P}_s}}{\sqrt{d^m_{SD}}} h_{SD} x_{SD} + n_d,&N_J>1,\\
\end{aligned}
\right.
\end{equation}
where $x_{SD}$ is the normalized information signal from the source in an IT block
and $n_d$ is the AWGN at the destination with variance $\sigma_d^2$.
Note that for $N_J>1$, the received signal is free of jamming, because the jamming signal is transmitted into the null space of $\bm h_{JD}$.

Similarly, the received signal at the eavesdropper, $y_E$, is given by
\begin{equation} \label{ye}
y_E\!=\!\!
\left\lbrace \!\!
\begin{aligned}
&\frac{\sqrt{\mathcal{P}_s}}{\!\!\sqrt{d^m_{SE}}}  h_{SE} x_{SD}
\!+ \!\frac{\sqrt{\mathcal{P}_J}}{\!\!\!\!\sqrt{d^m_{JE}}}  h_{JE}  x_{JD}
\!+\!n_e,& \!\!\!N_J\!=\!1,\\
&\frac{\sqrt{\mathcal{P}_s}}{\sqrt{d^m_{SE}}} h_{SE} x_{SD}
+ \frac{\sqrt{\mathcal{P}_J}}{\sqrt{d^m_{JE}}}  \bm h_{JE} \frac{\bm W \bm v}{\sqrt{N_J-1}}
\!+\!n_e,& \!\!\!N_J\!>\!1,
\end{aligned}
\right.
\end{equation}
where $n_e$ is the AWGN at the eavesdropper which we have assumed to be $0$ as a worst-case scenario.
\par
%
%
From \eqref{yd},
the SINR at the destination is
\begin{equation} \label{SNR_D}
\gamma_d=\left\lbrace
\begin{aligned}
&\frac{\frac{\mathcal{P}_s}{d^m_{SD}} \vert h_{SD}\vert ^2}{ \sigma_d^2 + \frac{\mathcal{P}_J}{d^m_{JD}}\vert h_{JD}\vert ^2},&N_J=1\\
&\frac{\mathcal{P}_s \vert h_{SD}\vert ^2}{d^m_{SD} \sigma_d^2},&N_J>1
\end{aligned}
\right.
\end{equation}
Hence the capacity of $S \rightarrow D$ link is given as $C_d = \log_2 \left(1+\gamma_d\right)$. 

Since $\lvert h_{SD} \rvert$ and $\lvert h_{JD} \rvert$ are Rayleigh distributed, $\lvert h_{SD} \rvert^2$ and $\lvert h_{JD} \rvert^2$ are exponential distributed and $\gamma_d$ has the cumulative distribution function (cdf) as
\begin{equation}\label{eqgammad}
F_{\gamma_d}\left(x\right)=\left\lbrace
\begin{aligned}
&1-\frac{e^{-\frac{x}{\rho_d}}}{1+\varphi x},& N_J=1,\\
&1-e^{-\frac{x}{\rho_d}}, & N_J >1,
\end{aligned}
\right. 
\end{equation}
where
\begin{equation} \label{varphi}
\varphi = \frac{\mathcal{P}_J}{\mathcal{P}_s}\frac{d^m_{SD}}{d^m_{JD}}.
\end{equation}
For convenience, we define the SNR at the destination (without jamming noise) as
\begin{equation} \label{m_rho_d}
\rho_d \triangleq \frac{\mathcal{P}_s}{d^m_{SD} \sigma^2_d}.
\end{equation}

From \eqref{ye}, the SINR at the eavesdropper is
\begin{equation} \label{SNR_E}
\gamma_e=\left\lbrace
\begin{aligned}
&\frac{1}{\phi}\frac{\vert h_{SE}\vert ^2}{\vert h_{JE} \vert^2}, &N_J=1,\\
&\frac{1}{\phi}\frac{\vert h_{SE}\vert ^2}{\frac{\Vert \bm h_{JE} \mathbf{W}\Vert^2}{N_J-1}}, &N_J>1,
\end{aligned}
\right.
\end{equation}
where
\begin{equation} \label{phi}
\phi=\frac{\mathcal{P}_J}{\mathcal{P}_s} \frac{d^m_{SE}}{d^m_{JE}}.
\end{equation}
Hence, the capacity of $S \rightarrow E$ link is given as $C_e = \log_2 \left(1+\gamma_e\right)$. 
Using the fact that $h_{SE}$, $h_{JE}$ and the entries of $\bm h_{JE} \mathbf{W}$ 
are independent and identically distributed (i.i.d.) complex Gaussian variables, 
from \cite{XiZhang}, $\gamma_e$ has the probability density function (pdf) as
\begin{equation}\label{eqgammae}
f_{\gamma _e}\left(x\right)=
\left\lbrace
\begin{aligned}
&\phi \left(\frac{1}{\phi x + 1}\right)^{2},&N_J=1,\\
&\phi \left(\frac{N_J-1}{\phi x + N_J -1}\right)^{N_J}, &N_J>1.
\end{aligned}
\right.
\end{equation}
Using the pdf of $\gamma_e$ in \eqref{eqgammae}, the secrecy outage probability defined in \eqref{p_so_first} can be evaluated.
\vspace*{-6pt}
\subsection{Information Transmission Probability}
Focusing on the long-term behavior, we analyze the proposed secure communication protocol and derive the information transmission probability $p_{tx}$, which in turn gives the throughput in \eqref{defi_B}. 
Note that $p_{tx}$ is the probability of an arbitrary block being used for IT. 
%
As discussed in the last section, the communication process falls in either energy accumulation or energy balanced case. 
Thus, $p_{tx}$ will have different values for the two different cases. 
First we mathematically characterize the condition of being in either case in the lemma below. 
\vspace*{-0.1cm}
\begin{lemma} \label{L1}
	\textnormal{The communication process with the proposed secure communication protocol leads to energy accumulation if 
	\vspace*{-5pt}
	\begin{equation}
	\frac{p_{co}}{1-p_{co}}  > \frac{\mathcal{P}_J T}{\rho_J}
	\end{equation}
	is satisfied. 
	Otherwise, the communication process is energy balanced.}
\end{lemma}
\begin{proof}
	See Appendix A.
\end{proof}
Using Lemma \ref{L1}, we can find the general expression for $p_{tx}$ as presented in Theorem \ref{P1} below. 
\begin{theorem} \label{P1}
	\textnormal{The information transmission probability for the proposed secure communication protocol is given by
	\vspace*{-5pt}
	\begin{equation} \label{p_tx_full}
	p_{tx} = 
	\frac{1}
	{1+\max \left\lbrace \frac{\mathcal{P}_J T}{ \rho_J},
		{\frac{p_{co}}{1-p_{co}}}
		\right\rbrace},	
	\end{equation}
	where 
	\begin{equation} \label{p_new_back}
	p_{co}=\left\lbrace
	\begin{aligned}
	&1-\frac{e^{-\frac{2^{R_t}-1}{\rho_d}}}{1+\frac{\mathcal{P}_J}{\mathcal{P}_s}\frac{d^m_{SD}}{d^m_{JD}}\left(2^{R_t}-1\right)}, &N_J=1,\\
	&1-e^{-\frac{2^{R_t}-1}{\rho_d}}, &N_J>1.
	\end{aligned}
	\right.
	\end{equation}}
\end{theorem}
\begin{proof}
	We first model the communication process in both energy accumulation and energy balanced cases as Markov chains and show the ergodicity of the process. 
	This then allows us to derive the stationary probability of a block being used for IT either directly or by using time averaging. 
	The detailed proof can be found in Appendix {B}.
\end{proof}
\par
By substituting \eqref{p_tx_full} in \eqref{defi_B}, we obtain the achievable throughput of the proposed protocol. 
%
\section{Optimal Design for Throughput}
In the last section, we derived the achievable throughput with given design parameters.
Specifically the jamming power $\mathcal{P}_J$ is a design parameter of the protocol.
A different value of $\mathcal{P}_J$ results in a different impact on the eavesdropper's SINR,
hence leads to different secrecy outage probability defined in \eqref{p_so_first}. 
Also the rate parameters of the wiretap code, $R_t$ and $R_s$, affect the secrecy outage probability. 
Hence, it is interesting to see how one can optimally design these parameters to maximize the throughput while keeping the secrecy outage probability below a prescribed threshold.
In this section, we present such an optimal fixed-rate design of the proposed secure communication protocol. 
The optimization is done offline, hence does not increase the complexity of the proposed protocol.
\vspace*{-10pt}
\subsection{Problem Formulation}
We consider the optimal secure communication design as follows:
\begin{equation} \label{my_first_opti}
\begin{aligned}
&\max_{\mathcal{P}_J, R_t,R_s} \ \ \ \pi \\
&\text{s.t.} \ \ \  p_{so} \leq \varepsilon,\ 
 \mathcal{P}_J \geq 0,\ 
 R_t \geq R_s \geq 0,
\end{aligned}
\end{equation}
where $\varepsilon$ is the secrecy outage probability constraint. This design aims to maximize the throughput with the constraint on the secrecy outage probability.
\par
From \eqref{p_so_first}, the secrecy outage probability should meet the requirement that
\begin{equation} \label{eqpso}
p_{so} =\mathbb{P}\left\lbrace R_t-R_s < \log_2\left(1+\gamma_e\right)  \right \rbrace \leq \varepsilon.
\end{equation}
By substituting \eqref{eqgammae} into \eqref{eqpso}, 
and after some simplification manipulations,
the jamming power $\mathcal{P}_J$ should satisfy the condition
\begin{equation} \label{secrecy_constraint}
\mathcal{P}_J
\geq
\mathcal{\hat{P}}_J
\triangleq
\left\lbrace
\begin{aligned}
&\mathcal{P}_s \frac{d^m_{JE}}{d^m_{SE}}\frac{\left(\varepsilon^{-1}-1\right)}{2^{R_t-R_s}-1},&N_J=1,\\
&\mathcal{P}_s \frac{d^m_{JE}}{d^m_{SE}}\frac{\left(N_J-1\right)\left(\varepsilon^{-\frac{1}{N_J-1}}-1\right)}{2^{R_t-R_s}-1},&N_J>1.
\end{aligned}
\right.
\end{equation}
From \eqref{p_tx_full}, we can see that $\pi$ decreases with $\mathcal{P}_J$. 
Thus, the maximum $\pi$ is obtained when
\begin{equation} \label{P_J}
\mathcal{P}^{\star}_J = \mathcal{\hat{P}}_J.
\end{equation}
\par
The jammer harvests energy from the source in each PT block. 
The dynamically harvested and accumulated energy at the jammer must exceed $\mathcal{P}^{\star}_J T$, before it can be used to send jamming signal with power $\mathcal{P}^{\star}_J$.

\par
Substituting \eqref{P_J} and \eqref{p_tx_full}, into \eqref{defi_B}, 
the throughput with optimal jamming power $\mathcal{P}^{\star}_J$ satisfying the secrecy outage constraint of $p_{so} \leq \varepsilon$, is given by \eqref{NAE_N=1} shown at the top of next page.
\begin{figure*}[!t]
	\normalsize
		\begin{equation} \label{NAE_N=1}
		\pi\left(\mathcal{P}^{\star}_J\right)
		=\left\lbrace
		\begin{aligned}
		&\frac{R_s}
		{1\!+\!
			\max\left\lbrace
			\underbrace{\frac{d^m_{SJ}}{\eta} \frac{d^m_{JE}}{d^m_{SE}}\frac{\left(\varepsilon^{-1}\!-1\!\right)}{2^{R_t-R_s}\!-\!1}}
			_{\left(a\right)},
			\underbrace{e^{\frac{\left(2^{R_t}-1\right)}{\rho_d}}\left(1\!+\!
				\frac{d^m_{JE}}{d^m_{SE}} \frac{d^m_{SD}}{d^m_{JD}} \frac{\left(\varepsilon ^{-1}-1\right)}{2^{R_t-R_s}-1}
				\left(2^{R_t}\!-\!1\right)\right)\!-1\!}_{\left(b\right)}
			\right\rbrace},\quad &N_J =1,\\
		&\frac{R_s}
		{1+\max\left\lbrace
			\underbrace{\frac{d^m_{SJ}}{N_J \eta} \frac{d^m_{JE}}{d^m_{SE}}\frac{\left(N_J-1\right)\left(\varepsilon^{-\frac{1}{N_J-1}}-1\right)}{2^{R_t-R_s}-1}}_{\left(a\right)},
			\underbrace{e^{\frac{2^{R_t}-1}{\rho_d}}-1}_{\left(b\right)}
			\right\rbrace}, \ \ \ \ \ \ \ \ \ \ \ \ \ \ \ \ &N_J > 1.
		\end{aligned}
		\right.\\
		\end{equation}
	\hrulefill
	\vspace*{-0.4cm}
\end{figure*}

Note that the terms $\left(a\right)$ and $\left(b\right)$ in \eqref{NAE_N=1} are the terms $\frac{\mathcal{P}_J T}{\rho_J}$ and $\frac{p_{co}}{1-p_{co}}$ in Lemma \ref{L1}, respectively.
Thus, if we choose $\left(R_t,R_s\right)$ to make $\left(a\right) > \left(b\right)$, the communication process is energy balanced; while if $\left(R_t,R_s\right)$ make $\left(a\right) \leq \left(b\right)$, the communication process leads to energy accumulation.
For analytical convenience, we define three 2-dimension rate regions:
\begin{equation}
\mathcal{D}_1 \triangleq  \left\lbrace \left(R_t,R_s\right) \vert (a) < (b), R_t\geq R_s \geq 0 \right\rbrace,
\end{equation}
\begin{equation}
\mathcal{\hat{D}}\  \triangleq  \left\lbrace \left(R_t,R_s\right) \vert (a) = (b), R_t\geq R_s \geq 0 \right\rbrace,
\end{equation}
\begin{equation}
\mathcal{D}_2 \triangleq  \left\lbrace \left(R_t,R_s\right) \vert (a) > (b), R_t\geq R_s \geq 0 \right\rbrace,
\end{equation}
where rate region $\mathcal{\hat{D}}$ denotes the boundary between regions $\mathcal{D}_1$ and $\mathcal{D}_2$.
From the discussion above, if $\left(R_t,R_s\right) \in \mathcal{D}_1$, the communication process leads to energy accumulation, while if $\left(R_t,R_s\right) \in \mathcal{D}_2 \cup \mathcal{\hat{D}}$, it is energy balanced.

Using \eqref{NAE_N=1}, the optimization problem in \eqref{my_first_opti} can be rewritten as
\begin{equation} \label{opt_new}
\begin{aligned}
\max\limits_{ R_t, R_s} \ &\ \pi\left(\mathcal{P}^{\star}_J\right) \\
\mathrm{s.t.}  \ R_t \geq & R_s \geq 0.
\end{aligned}
\end{equation}
The optimization problem in \eqref{opt_new} can be solved with global optimal solution. 
The solution for $N_J =1$ and $N_J > 1$ are presented in the next two subsections.
\subsection{Optimal Rate Parameters with Single-Antenna Jammer}
\begin{proposition} \label{P3}
	\textnormal{When $N_J = 1$, the optimal $R_t$ and $R_s$ can be obtained by using the following facts:\\
	{\textbf{IF}} $\left(R^{\star}_t, R^{\star}_s\right) \in \mathcal{D}_1$,
	i.e., the case of energy accumulation,
	$R^{\star}_s$ is the unique root of equation (monotonic increasing on the left side):
	\begin{equation} \label{Middel_4_new}
	k_2\left(2^{R_s} + \frac{2^{R_s} -1}{\xi}\right) \left({R_s} \ln2-1 +\frac{{R_s}\ln2}{\xi} \right)  =1,
	\end{equation}
	and $R^{\star}_t$ is given by 
	\vspace*{-5pt}
	\begin{equation} \label{P3_Rb}
	R^{\star}_t = R^{\star}_s +\log_2\left( 1 + \xi^{\star}\right),
	\end{equation}
	where
	\begin{equation} \label{y_previous}
	\xi\!=\!\!
	\frac{1}{2}\!
	\left(\!\!
	-\frac{k_2 \left(2^{R_s} \!-\!1 \right)}{1\!+\!k_2 2^{R_s}}
	\!+\!\!\left(\!\!
	\left(\!\frac{k_2 \left(2^{R_s} \!-\!1 \right)}{1\!+\!k_2 2^{R_s}}\right)^{\!\!2}
	\!\!+\!\!\frac{4 \rho_d k_2 \left( 1 \!-\!\frac{1}{2^{R_s}} \right)}{1\!+\!k_2 2^{R_s}}\!
	\right)^{\!\!\!\frac{1}{2}}\!
	\right)\!\!,
	\end{equation}
	and $\xi^{\star}$ is obtained by taking $R^{\star}_s$ into \eqref{y_previous}.\\
	{\textbf{ELSE}},  $\left(R^{\star}_t, R^{\star}_s\right) \in \mathcal{\hat{D}}$,
	i.e., the energy balanced case,
	and	$R^{\star}_t$ is the root of following equation which can be easily solved by a linear search:
	\begin{equation} \label{P3_U}
	\zeta' \left(
	\frac{1+\frac{k_1}{\zeta}}{\ln 2 \left(1+\zeta\right)}
	- \frac{k_1 \left(R_t - \log_2 \left(1+\zeta\right)\right)}{\zeta^2}
	\right)=1,
	\end{equation}
	where
	\begin{equation} \label{U_new}
	\zeta = \frac{k_1 - k_2 e^{\frac{2^{R_t}-1}{\rho_d}} \left(2^{R_t}-1\right)}{e^{\frac{2^{R_t}-1}{\rho_d}} -1},
	\end{equation}
	\begin{equation}
	\zeta'\!\! =\!\! \frac{\ln 2 \ e^{\frac{2^{R_t}-1}{\rho_d}}}{\left(e^{\frac{2^{R_t}-1}{\rho_d}} \!-\!1\right)^2}\!
	\left(
	k_2 \ 2^{R_t} \left(1\!+\!\frac{1}{\rho_d} \!- \!e^{\frac{2^{R_t}-1}{\rho_d}} \right) - \frac{k_1 \!+\!k_2}{\rho_d}
	\right),
	\end{equation}
	\begin{equation} \label{def_k1}
	k_1 = \frac{d^m_{SJ}}{\eta} \frac{d^m_{JE}}{d^m_{SE}} {\left(\varepsilon^{-1}-1\right)},
	\end{equation}
	\begin{equation} \label{def_k2}
	k_2= \frac{d^m_{JE}}{d^m_{SE}} \frac{d^m_{SD}}{d^m_{JD}}  {\left(\varepsilon ^{-1}-1\right)},
	\end{equation}	
	and $R^{\star}_s = R^{\star}_t - \log_2\left(1+\zeta^{\star}\right)$, where $\zeta^{\star}$ is calculated by taking $R^{\star}_t$ into \eqref{U_new}.}
\end{proposition}
\begin{proof}
	See Appendix C.
\end{proof}
Note that the optimal $\left(R_t,R_s\right)$ never falls in region $\mathcal{D}_2$. 
This is because the throughput in $\mathcal{D}_2$ increases towards the boundary of $\mathcal{D}_1$ and $\mathcal{D}_2$, that is $\hat{\mathcal{D}}$. 
The detailed explanation is given in Appendix C. 
\par
Proposition \ref{P3} can then be used to obtain the optimal values of $R_t$ and $R_s$ as follows. 
We firstly assume the optimal results are in the region $\mathcal{D}_1$, thus, $R_s$ and $R_t$ can be obtained by equation \eqref{Middel_4_new} and \eqref{P3_Rb}. Then, we check whether the results are really in $\mathcal{D}_1$. 
If they are, we have obtained the optimal results. 
If not, we solve equation \eqref{P3_U} to obtain the optimal $R_t$ and $R_s$.

\subsubsection{High SNR Regime}
We want know whether we can largely improve throughput by increasing 
the transmit power at the source,
$\mathcal{P}_s$, thus, we consider the high SNR regime.
Note that we have defined SNR at the destination (without the effect of jamming noise) as $\rho_d$ in \eqref{m_rho_d}.
\begin{corollary} \label{C1}
	\textnormal{When $N_J = 1$ and the SNR at the destination is sufficiently high, the asymptotically optimal rate parameters and an upper bound of throughput are given by
\begin{subequations}
		\begin{align} \label{}
	\tilde{R}^{\star}_s &= \frac{1 + \mathrm{W}_0 \left(\frac{1}{e k_2}\right)}{\ln 2},
	\\
	\tilde{R}^{\star}_t &= \tilde{R}^{\star}_s + \log_2 \left(1+ \tilde{\xi}^{\star}\right),
	\\
	\tilde{\pi}^{\star}
	&= \frac{\mathrm{W}_0 \left(\frac{1}{e k_2}\right)}{\ln 2},
	\end{align}
\end{subequations}
	where $k_2$ is defined in \eqref{def_k2},
	\vspace*{-6pt}
	\begin{equation}
	\tilde{\xi}^{\star} =
	\left(\frac{\rho_d k_2 \left(1 - \frac{1}{2^{\tilde{R}^{\star}_s}}\right)}{1 + k_2 2^{\tilde{R}^{\star}_s}}\right)
	^{\frac{1}{2}},
	\end{equation}
	and $\mathrm{W}_0 \left(\cdot \right)$ is the principle branch of the Lambert W function~\cite{lambert}.}
\end{corollary}
\par
\begin{proof}
	See Appendix D.
\end{proof}
\par

\textit{Remarks:} \par
\begin{enumerate}[i)]
\item  The upper bound of throughput 
implies that one cannot effectively improve the throughput by further increasing $\mathcal{P}_s$ when the SNR at the destination is already high. 
\item 
It can be checked that when $\mathcal{P}_s$ is sufficiently high, 
the optimized communication process leads to energy accumulation. 
Intuitively, this implies that when the available harvested energy is very large, 
the jammer should store a significant portion of the harvested energy in the battery rather than fully using it, because too much jamming noise can have adverse impact on SINR at the destination in this single-antenna jammer scenario. 
This behavior will also be verified in Section VI, Fig. \ref{fig:compare4}.
\end{enumerate}

\subsection{Optimal Rate Parameters with Multiple-Antenna Jammer}
\begin{proposition} \label{P4}
	\textnormal{When $N_J > 1$, the optimal $R_s$ and $R_t$ are in region $\mathcal{\hat{D}}$ which also means that the optimal communication process is in the energy balanced case, and the optimal values are given by
	\begin{equation} \label{}
	\begin{aligned}
	R^{\star}_t&=\log_{2} z^{\star},\\
	R^{\star}_s&=\log_{2} \frac{z^{\star}}{1+\frac{M}{e^{\frac{z^{\star}-1}{\rho_d}}-1}},
	\end{aligned}
	\end{equation}
	where $z^{\star}$ is calculated as the unique solution of
	\begin{equation} \label{equal_previous}
	\frac{\rho_d}{z}-\ln z
	+ \ln \!\left(\!1\!+\!\frac{M}{e^{\frac{z-1}{\rho_d}}\!-\!1}\!\right)
	+ \frac{M e^{\frac{z-1}{\rho_d}}}
	{\left(e^{\frac{z-1}{\rho_d}}\!-\!1\right)^2 \!\!+\! M \left(e^{\frac{z-1}{\rho_d}}\!-\!1\!\right)}
	\!= \!0,
	\end{equation}
	and
	\begin{equation} \label{my_M}
	M=
	\frac{d^m_{SJ}}{N_J \eta} \frac{d^m_{JE}}{d^m_{SE}}{\left(N_J-1\right)\left(\varepsilon^{-\frac{1}{N_J-1}}-1\right)}.
	\end{equation}}
\end{proposition}
\begin{proof}
	See Appendix E.
\end{proof}

We can see that the left side of \eqref{equal_previous} is a monotonic decreasing function of $z$.
Thus, $z$ can be easily obtained by using numerical methods.
\par

\subsubsection{High SNR Regime}
Similar to the single-antenna jammer case, we are interested in whether increasing the source transmission power $\mathcal{P}_s$, is an effective way of improving throughput. 
Hence the high SNR regime is considered:

\begin{corollary} \label{C3}
	\textnormal{When $N_J > 1$ and the SNR at the destination is sufficiently high,
	the asymptotically optimal rate parameters and an upper bound of throughput are given by
	\begin{subequations}
		\begin{align}
		\label{slow_N>1_Op_R_b}
		\tilde{R}^{\star}_t &= \log_2\left(2\rho_d\right) -\log_2\left(\mathrm{W_0}\left(2 \rho_d\right)\right),
		\\
		\label{slow_N>1_Op_R_s}
		\tilde{R}^{\star}_s &
		= \frac{2 \mathrm{W_0}\left(2 \rho_d\right)}{\ln 2} - \log_2\left(M \rho_d\right),
		\\
		\label{slow_N>1_Op_ans}
		\tilde{\pi}^{\star} &
		= \tilde{R}^{\star}_s,
		\end{align}
	\end{subequations}
	where
	$
	z^{\star} = \frac{2 \rho_d}{\mathrm{W_0}\left(2 \rho_d\right)}
	$ 
	and 
	$M$ is defined in \eqref{my_M}.}
\end{corollary}

\begin{proof}
See Appendix F.	
\end{proof}\par

\par
\textit{Remarks:} 
\begin{enumerate}[i)]
	\item The throughput will always increase with increasing transmit power $\mathcal{P}_s$ (because $\rho_d$ increases as $\mathcal{P}_s$ increases). 
	This is in contrast to the single-antenna jammer case, because the multi-antenna jamming method only interferes the $S \rightarrow E$ link.
	\item 
	From Proposition \ref{P4} and Corollary \ref{C3}, when $\mathcal{P}_s$ is sufficiently large, the optimized communication process is still energy balanced, which is different from the single-antenna jammer scenario. 
Intuitively, 
storing extra energy is not a good choice, 
because we can always use the accumulated energy to jam at the eavesdropper without affecting the destination, which in turn improves the throughput.
\end{enumerate}

\subsubsection{Large $N_J$ Regime}
We also want to know that whether we can largely improve the throughput by increasing the number of antennas at the jammer.
\begin{corollary} \label{C2}
	\textnormal{In large $N_J$ scenario,
	the asymptotically optimal rate parameters and an upper bound of throughput are given by
\begin{subequations}
		\begin{align}
	\tilde{R}^{\star}_t &= \frac{\mathrm{W}_0 \left(\rho_d\right)}{\ln 2},
	\\
	\tilde{R}^{\star}_s &= \log_2 \frac{e^{\mathrm{W}_0 \left(\rho_d\right)}}{1+\frac{M}{e^{\frac{e^{\mathrm{W}_0 \left(\rho_d\right)}-1}{\rho_d}}-1}},
	\\
	\tilde{\pi}^{\star} &= \frac{\mathrm{W}_0\left(\rho_d\right)}{\ln 2 \ e^{\frac{1}{\mathrm{W}_0\left(\rho_d\right)} - \frac{1}{\rho_d}}},
	\end{align}
\end{subequations}
	where $M$ is defined in \eqref{my_M}.}
\end{corollary}

\begin{proof}
See Appendix F.
\end{proof}\par

\par
\textit{Remark:}
 Corollary \ref{C2} gives an asymptotic upper bound on throughput for this protocol, thus, 
$\pi$ does not increase towards infinity with $N_J$. 
Intuitively, the throughput cannot always increase with $N_J$, because it is bounded by the  $S \rightarrow D$ channel capacity which is independent with $N_J$.

\section{Numerical Results}
In this section, we present numerical results to demonstrate the performance of the proposed secure communication protocol. We set the path loss exponent as $m = 3$ and the length of time block as $T = 1$ millisecond. 
{{We set the energy conversion efficiency as $\eta= 0.5$ \cite{ZhangHo,JuOld,XunZhou}.}
	{Note that the practical designs of rectifier  for RF-DC conversion achieve the value of $\eta$ between~$0.1$ and $0.85$~\cite{WPT_survey}.}
	{Such a range makes wireless energy harvesting technology meaningful. A rectifier design
	with }{$\eta < 0.1$ }{is unlikely to be used in practice.}}
We assume that the source, jammer, destination and eavesdropper
are placed along a horizontal line,  
and the distances are given by $d_{SJ} =25$ m, $d_{SE} =40$ m, $d_{SD} =50$ m, $d_{JE} =15$ m, $d_{JD} =25$ m, 
in line with \cite{dong}. 
Unless otherwise stated, we set $\sigma^2_d = -100$ dBm, 
and the secrecy outage probability requirement $\varepsilon = 0.01$. 
We do not specify the bandwidth of communication, hence the rate parameters are expressed in units of bit per channel use (bpcu).
\par
To give some ideas about the energy harvesting process at the jammer under this setting:
	When $N_J = 1$ and $\mathcal{P}_s = 30$ dBm,
	the average power that can be harvested (after RF-DC conversion) is $-15$ dBm,
	thus, the overall energy harvesting efficiency (i.e., the ratio between the harvested power at the jammer and the transmit power at the source) is $(-15\text{ dBm})/(30\text{ dBm})$ $\approx 3\times 10^{-5}$.
	Note that, although the average harvested power at the jammer is relatively small, a small jamming power is sufficient to achieve good secure communication performance.
	For instance, the optimal jamming power under this setting is only $-13$ dBm based on the analytical results in Section~V.
	In order to transmit the jamming signal at the optimal power of $-13$ dBm with the average harvested power of $-15$ dBm, roughly $61\%$ of time is used for charging and $39\%$ of time is used for secure communication with jamming.	
\subsection{Energy Accumulation and Energy Balanced Cases}
\begin{figure}[t]
	\renewcommand{\captionlabeldelim}{ }
	\renewcommand{\captionfont}{\small} \renewcommand{\captionlabelfont}{\small}
	\centering
	\includegraphics[scale=0.62]{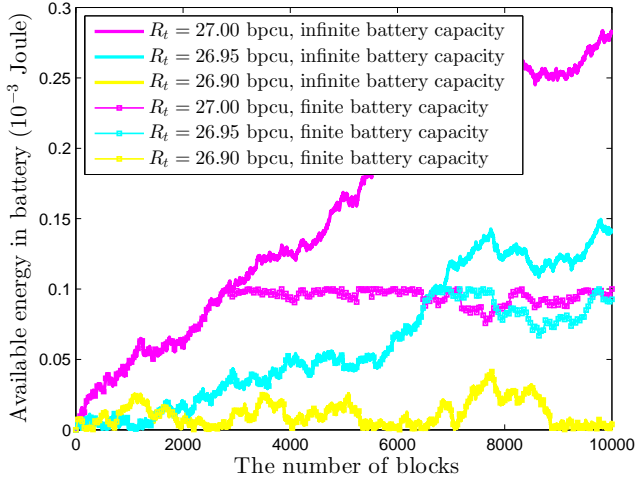}
	\vspace*{-0.5cm}
	\caption{Available energy in battery during the communication process.}	
	\label{fig:energy state}
	\vspace*{-0.3cm}
\end{figure}

Fig. \ref{fig:energy state} shows the simulation results on the available energy in the battery in the communication process.
The jammer has 8 antennas ($N_J = 8$) and the target jamming power is $\mathcal{P}_J = 0$ dBm. 
The source transmit power is $\mathcal{P}_s = 30$ dBm. 
Thus, the energy consumption in one IT block at the jammer, $\mathcal{P}_J T$, is $10^{-6}$ Joule, and the average harvested energy in one PT block, $\rho_J$, is $2.56 \times 10^{-7}$ Joule.
From Lemma \ref{L1} and \eqref{p_new_back}, when $\frac{p_{co}}{1-p_{co}}> \frac{\mathcal{P}_J T}{\rho_J}$ which means $R_t > 26.92$ bpcu, the communication process leads to energy accumulation, while if $R_t \leq 26.92$ bpcu, it is the energy balanced.
\par
{First, we focus on the curves with infinite battery capacity.} We can see that when $R_t =26.9$ bpcu, the available energy goes up and down, but does not have the trend of energy accumulation. 
Thus, the communication process is energy balanced. 
When $R_t = 26.95$ and $27$ bpcu, the available energy grows up, 
and the communication process leads to energy accumulation. 
Therefore, the condition given in Lemma \ref{L1} is verified.
\par
{{{In Fig. \ref{fig:energy state}, we also plot a set of simulation results with finite battery capacity as}}
	{${E_{\textrm{max}} = 0.1\times 10^{-3}}$ {Joule}}.
	{As we can see, for the energy accumulative cases, i.e., $R_t = 26.95$ and $27.00$ bpcu, 
		the energy level stays near the battery capacity ($0.1\times 10^{-3}$ Joule) after experienced a sufficient long time, which is much higher than the required jamming energy level $\mathcal{P}_J T = 10^{-6}$ Joule.
		Therefore, in practice, having a finite battery capacity has hardly any effect on the communication process, as compared with infinite capacity.}
}

\subsection{Rate Regions with Single-Antenna Jammer}
Fig. \ref{fig:compare4} plots the throughput in \eqref{NAE_N=1} with different $R_t$ and $R_s$  in the single-antenna jammer scenario. 
In Fig. \ref{fig:N_1_low}, we set $\mathcal{P}_s = 0$ dBm. The optimal rate parameters $\left(R^{\star}_t, R^{\star}_s\right)$ are obtained in the region $\hat{\mathcal{D}}$, which is the boundary of $\mathcal{D}_1$ and $\mathcal{D}_2$. 
This implies that the optimized communication process is energy balanced. 
In Fig. \ref{fig:N_1_high}, we increase $\mathcal{P}_s$ to $30$ dBm. The optimal rate parameters $\left(R^{\star}_t, R^{\star}_s\right)$ are obtained in the region ${\mathcal{D}}_1$. 
This implies that the optimized communication process leads to energy accumulation. 
This observation agrees with the remarks after Corollary \ref{C1} regarding the optimal operating point when the SNR at the destination is sufficiently large. 

\begin{figure}[t]
	\renewcommand{\captionlabeldelim}{ }
	\renewcommand{\captionfont}{\small} \renewcommand{\captionlabelfont}{\small}
	\centering
	  \hspace{-0.5in}
	\subfigure[$N_J=1$, $\mathcal{P}_s = 0$ dBm]{
		\label{fig:N_1_low} 
		\includegraphics[scale=0.37]{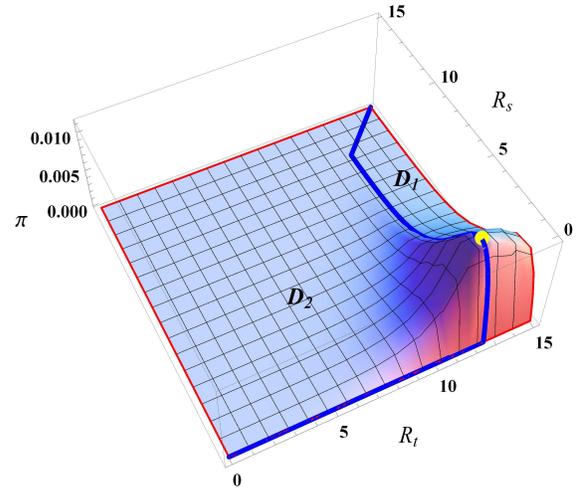}}
     \hspace{2.5in}
	\subfigure[$N_J=1$, $\mathcal{P}_s = 30$ dBm]{
		\label{fig:N_1_high} 
		\includegraphics[scale=0.37]{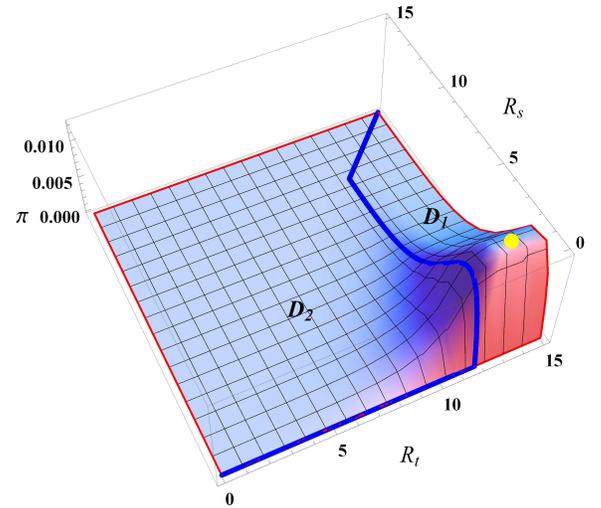}}	
	  \hspace{2.5in}	
	\caption{Optimal rate parameters for $N_J=1$.}
	\label{fig:compare4} 
	\vspace*{-15pt}
\end{figure}
\vspace*{-0.5cm}
\begin{figure}[t]
	\renewcommand{\captionlabeldelim}{ }
	\renewcommand{\captionfont}{\small} \renewcommand{\captionlabelfont}{\small}
	\centering
	\includegraphics[scale=0.62]{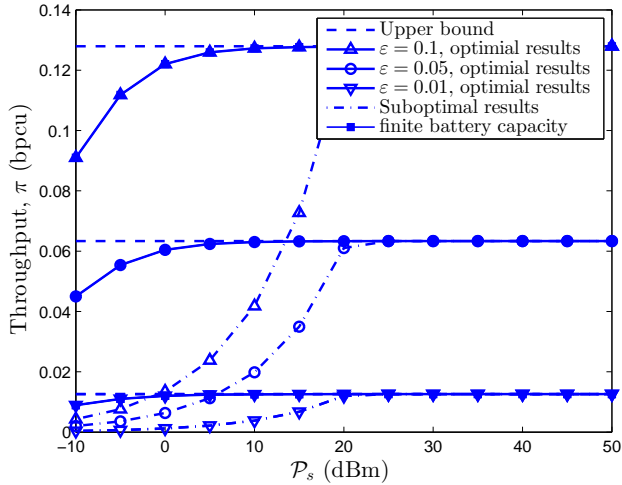}
	\vspace*{-0.5cm}
	\caption{Throughput versus source transmit power $\mathcal{P}_s$ for $N_J = 1$.}
	\label{fig:N0}
	\vspace*{-0.5cm}
\end{figure}

\subsection{Throughput Performance with Single-Antenna Jammer}
Fig. \ref{fig:N0} plots the throughput with optimal designs given by Proposition \ref{P3}. We also include the suboptimal performance which is achieved by using the asymptotically optimal rate parameters in Corollary \ref{C1}, as well as the upper bound on throughput in Corollary \ref{C1}.
\par
{First, we focus on the curves with infinite battery capacity.} We can see that when $\mathcal{P}_s = 5$~dBm, the optimal throughput almost reaches the upper bound. 
Also we can see that when $\mathcal{P}_s < 20$~dBm, the suboptimal performance has a large gap with the optimal one, while when $\mathcal{P}_s > 20$~dBm, the suboptimal performance is very close to the optimal one. 
\par
{{{In Fig. \ref{fig:N0}, we also plot a set of simulation results with finite battery capacity as}}
	{${E_{\textrm{max}}= 0.1}$ \\${ \times 10^{-3}}$ {Joule}}.					
	{It is easy to see that our analytical results for infinite battery capacity fit very well with the simulation results for finite battery capacity.
		Therefore, a practical finite battery capacity have negligible effect on the performance of the protocol,
		and our analysis are valid in the practical scenario.}
}

\subsection{Rate Regions with Multiple-Antenna Jammer}
Fig. \ref{fig:compare4_new} plots the throughput in \eqref{NAE_N=1} with different $R_t$ and $R_s$ in the multiple-antenna jammer scenario.
In Fig. \ref{fig:N_8_low} and Fig. \ref{fig:N_8_high}, we set $\mathcal{P}_s = 0$~dBm and $30$~dBm, respectively. 
The optimal rate parameters $\left(R^{\star}_t, R^{\star}_s\right)$ are both obtained in the region $\hat{\mathcal{D}}$. 
This implies that the optimized communication process is energy balanced, which agrees with the remarks after Corollary \ref{C3}. 

\begin{figure}[t]
	\renewcommand{\captionlabeldelim}{ }
	\renewcommand{\captionfont}{\small} \renewcommand{\captionlabelfont}{\small}
	\centering
	\subfigure[$N_J=8$, $\mathcal{P}_s= 0$ dBm]{
		\label{fig:N_8_low} 
		\includegraphics[scale=0.37]{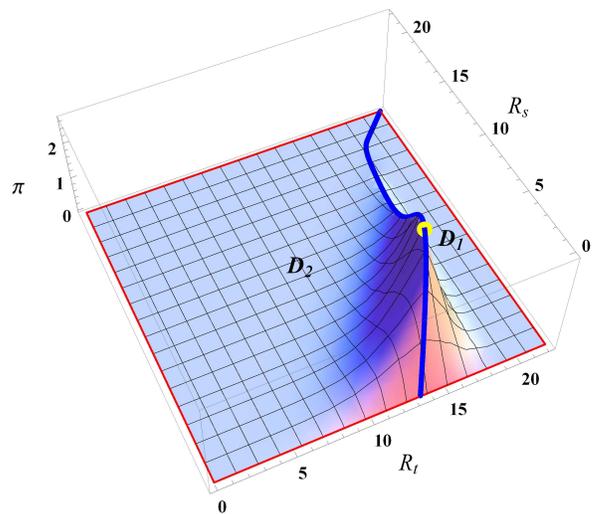}}	
	\subfigure[$N_J=8$, $\mathcal{P}_s= 30$ dBm]{
		\label{fig:N_8_high} 
		\includegraphics[scale=0.37]{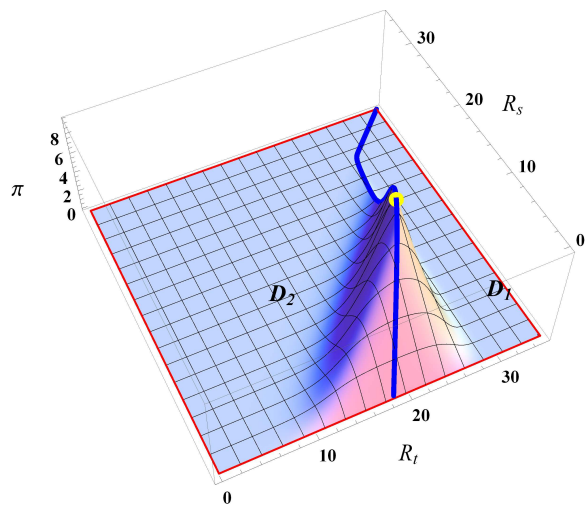}}			
	\caption{Optimal rate parameters for $N_J = 8$.}
	\label{fig:compare4_new} 
	\vspace*{-0.5cm}	
\end{figure}

\subsection{Throughput Performance with Multiple-Antenna Jammer}
Fig. \ref{fig:N-1_2} plots the optimal throughput from Proposition \ref{P4}. 
We also present the suboptimal performance which is achieved by the asymptotically optimal rate parameters obtained in Corollary \ref{C3}. 
We can see that the throughput increases with $\mathcal{P}_s$ unbounded. 
Also we can see that the suboptimal performance is reasonably good when $\mathcal{P}_s > 20$~dBm.

\begin{figure}[t]
	\renewcommand{\captionlabeldelim}{ }
	\renewcommand{\captionfont}{\small} \renewcommand{\captionlabelfont}{\small}
	\centering
	\subfigure[Throughput vs. source transmit power $\mathcal{P}_s$ for $N_J = 8$.]{
		\label{fig:N-1_2} 
		\includegraphics[scale=0.62]{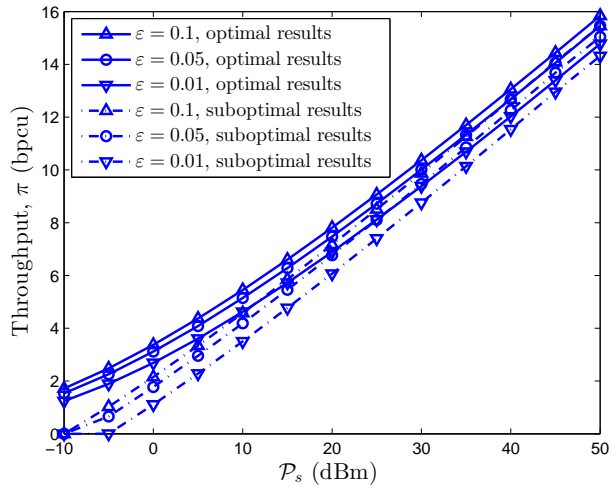}}
	\subfigure[Throughput vs. number of antennas at the jammer, $N_J \geq 2$.]{
		\label{fig:N-1} 
		\includegraphics[scale=0.62]{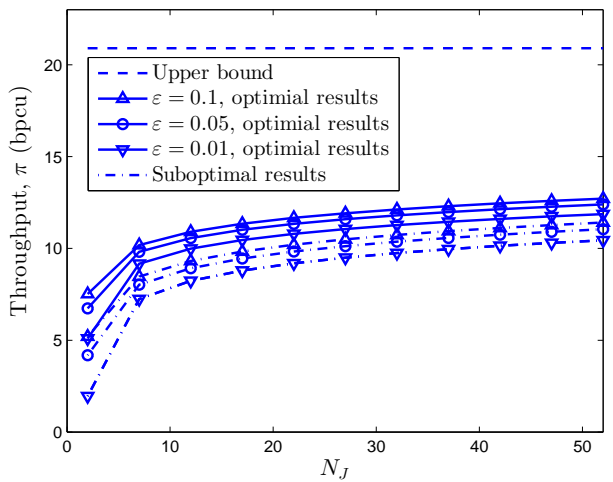}}	
	\vspace*{-0.2cm}
	\caption{Throughput for $N_J > 1$.}
	\label{fig:the new} 
	\vspace*{-0.5cm}	
\end{figure}

%

Fig. \ref{fig:N-1} plots the throughput achieved with the optimal design given in  Proposition \ref{P4} for different $N_J$. The source transmit power is $\mathcal{P}_s = 30$~dBm. 
We also include the suboptimal performance achieved by the asymptotically optimal rate parameters in the large $N_J$ regime (Corollary \ref{C2}) as well as the upper bound on throughput in Corollary \ref{C2}.

We can see that with the increment of $N_J$,
although theoretically the throughput is upper bounded as $N_J \rightarrow \infty$, the available throughput within practical range of $N_J$ is far from the upper bound. 
Hence, increasing $N_J$ is still an efficient way to improve the throughput with practical antenna size.
Also we can see that the suboptimal performance is acceptable but the gap from the optimal throughput performance is still noticeable.

\section{Conclusion}
In this paper, we investigated secure communication with the help from a wireless-powered jammer. 
We proposed a simple communication protocol and derived its achievable throughput with fixed-rate transmission. 
We further optimized the rate parameters to achieve the best throughput subject to a secrecy outage probability constraints. 
As energy harvesting and wireless power transfer become emerging solutions for energy constrained networks, this work has demonstrated how to make use of an energy constrained friendly jammer to enable secure communication without relying on an external energy supply. 
For future work, the protocol can be extended to include
more sophisticated adaptive transmission schemes,
such as variable power transmission with an average power constraint at the source.
Also these schemes can be generalized to multiple antennas at all nodes as well but with a certain constraint on the receiver noise level or the number of transmit/receive antennas at the jammer/eavesdropper (as needed in all physical layer security work).
We will explore these relevant problems in our further work. 
Also our design idea can be borrowed and apply other EH method, such as solar, vibration, thermoelectric, wind and even hybrid energy harvesting with several energy sources. However, apart from secure communication performance and EH efficiency, dimension requirements, implementation complexity, costs should be considered. 
Also our design idea can be borrowed and applied with other EH methods, such as solar, vibration, thermoelectric, wind, and even hybrid energy harvesting with several energy sources. However, apart from communication performance and EH efficiency, dimension requirements, implementation complexity and costs should also be taken into account in the design. 

\begin{appendices}
	
	\renewcommand{\theequation}{\thesection.\arabic{equation}} 
	\numberwithin{equation}{section}
	\vspace*{-0.2cm}
	\section{Proof of Lemma \ref{L1}}

		In one PT-IT cycle, once the available energy is higher than $\mathcal{P}_J T$, there will be $Y$ opportunistic PT blocks. 
		The probability of the discrete random variable $Y$ being $k$ is the probability that the successive $k$ blocks, suffer from connection outage of the $S \rightarrow D$ link, and the $\left(k+1\right)$th block does not have the $S \rightarrow D$ outage. 	
		Due to the i.i.d.  channel gains in different blocks, $Y$ follows a geometric distribution and the probability mass function (pmf) is given by
		\begin{equation} \label{Y=i}
		\mathbb{P}\left\lbrace Y=k \right\rbrace=p^k_{co} \left(1-p_{co}\right), k=0,1,....
		\end{equation}
		The mean value of $Y$ is given by
		\begin{equation} \label{EY}
		\mathbb{E}\left\lbrace Y \right\rbrace
		=\sum\limits_{k=0}^{\infty} k \mathbb{P}\left\lbrace Y=k \right\rbrace
		=\sum\limits_{k=0}^{\infty} k p^k_{co} \left(1-p_{co}\right)
		=\frac{p_{co}}{1-p_{co}}.
		\end{equation}

		As we have defined $\rho_J$ as the average harvested energy by one PT block, the average harvested energy by $Y$ opportunistic PT blocks in one PT-IT cycle is given by
		\begin{equation}
		\mathcal{E}_{Y}
		=\mathbb{E} \left\lbrace Y \right \rbrace \rho_J
		=\frac{p_{co}}{1-p_{co}} \rho_J.
		\end{equation}
		If the average harvested energy by opportunistic PT blocks in a PT-IT cycle is higher than the required energy, $\mathcal{P}_J T$, for jamming in one IT block, the communication process leads to energy accumulation. 
		Otherwise, we need dedicated PT blocks in some PT-IT cycles, and the communication process is energy balanced.
		Thus, we have the condition in Lemma 1.

\section{Proof of Theorem \ref{P1}}
We derive information transmission probability $p_{tx}$ in the following two cases.\par
\textbf{Energy Accumulation Case}:
In this case, there are no dedicated PT blocks.
We use a simple Markov chain with two states, IT and opportunistic PT, to model the communication process. 
When the fading channel of $S \rightarrow D$ link suffers connection outage, 
the block is in the opportunistic PT state, otherwise it is in the IT state. 
This Markov chain is ergodic since the fading channel of $S \rightarrow D$ link is i.i.d. between blocks. 
The information transmission probability is simply the probability that the $S \rightarrow D$ link does not suffer connection outage, hence we have
\begin{equation} \label{L2}
p_{tx} = 1-p_{co} = \frac{1}{1+\frac{p_{co}}{1-p_{co}}}.
\end{equation}

\textbf{Energy Balanced Case}:
In this case, the available energy at the jammer becomes directly relevant to whether a block is used for IT or PT. 
Following the recent works, such as \cite{Lee}, we model the energy state at the beginning/end of each time block as a Markov chain in order to obtain the information transmission probability. 
Since the energy state is continuous, we adopt Harris chain which can be treated as a Markov chain on a general state space (continuous state Markov chain). 

It is easy to show that this Harris chain is recurrent and aperiodic, because any current energy state can be revisited in some future block, 
and one cannot find any two energy states that the transition from one to the other is periodic. 
Therefore,
the Harris chain is ergodic \cite{Harris}, and there exists a unique stationary measure (stationary distribution), which means that
the stationary distribution of available energy at the beginning/end of each block exists. 
Thus, the stationary probability of a block being used for IT ($p_{tx}$) or PT exists. 
\par
Instead of deriving the stationary distribution of energy states, we use time averaging which makes use of the ergodic property, to calculate the information transmission probability $p_{tx}$ which is given by
\vspace*{-4pt}
 \begin{equation} \label{appen_p_tx}
 p_{tx}
 =
 \lim\limits_{\mathcal{N}_{total} \rightarrow \infty}
 \frac{\mathcal{N}_{IT}}{\mathcal{N}_{PT} + \mathcal{N}_{IT}}
=  \lim\limits_{\mathcal{N}_{total} \rightarrow \infty}
  \frac{1}{1+\mathcal{N}_{PT}/\mathcal{N}_{IT}}
 ,
 \end{equation}
where $\mathcal{N}_{IT}$ and $\mathcal{N}_{PT}$ denotes the number of IT and PT blocks in the communication process,
$\mathcal{N}_{total} \triangleq \mathcal{N}_{PT}+ \mathcal{N}_{IT}$.
By using the principle of conservation of energy (i.e., all the harvested energy in PT blocks  are used for jamming in IT blocks) and the law of large numbers, we have
\begin{equation} \label{N_ave}
\lim\limits_{\mathcal{N}_{total} \rightarrow \infty}
\frac{\mathcal{N}_{PT} \ \rho_J}{\mathcal{N}_{IT} \ \mathcal{P}_J T} = 1 ,
\end{equation}
where $\rho_J$ is the average harvested energy in one PT block defined in \eqref{my_rho} and $\mathcal{P}_J T$ is the energy used for jamming in one IT block. 
By taking \eqref{N_ave} into \eqref{appen_p_tx} the information transmission probability is given by
	\vspace*{-12pt}
\begin{equation} \label{L3}
p_{tx} = \frac{1}{1+\frac{\mathcal{P}_J T}{\rho_J}}.
\end{equation}
\par
\textbf{General Expression}: 
	Based on Lemma \ref{L1}, \eqref{L2} and \eqref{L3}, we can easily obtain the general expression for $p_{tx}$ as
	\vspace*{-5pt}
		\begin{equation}\label{p_tx_proof_1}
		p_{tx}
		= \frac{1}{1+ \max\left\lbrace \frac{\mathcal{P}_J T}{ \rho_J}, \frac{p_{co}}{1-p_{co}} \right\rbrace}.
		\end{equation}
	From \eqref{p_co}, we have,
	\begin{equation}\label{eqpco}
	p_{co}\! =\!\mathbb{P} \left\lbrace \log_2\left(1\!+\!\gamma_d\right) \!<\! R_t \right\rbrace
	\!=\!\mathbb{P}\left\lbrace \gamma_d \!<\! 2^{R_t}\!-\!1 \right\rbrace
	\!\!=\!F_{\gamma_d}\!\left(2^{R_t}\!-\!1\right)\!.
	\end{equation}
	By taking \eqref{eqgammad} into \eqref{eqpco}, we obtain the expression of $p_{co}$ in~\eqref{p_new_back}.

\section{Proof of Proposition \ref{P3}}	
	\textbf{Case I}: If optimal $\left(R_t,R_s\right) \in \mathcal{D}_1$,
	the optimization problem can be rewritten as	
	\begin{equation} \label{a<b}
	\max_{\left(R_t,R_s\right) \in \mathcal{D}_1} \ \
	\pi=\frac{R_s}
	{
		e^{\frac{\left(2^{R_t}-1\right)}{\rho_d}}
		\left(1+
		\frac{k_2 \left(2^{R_t} -1\right)}{2^{R_t-R_s}-1}
		\right)
	}.
	\end{equation}
	
	The optimal $\left(R_t,R_s\right)$ should satisfies
	$\frac{\partial \  \pi}{\partial \ \varsigma}=0$ and
	$\frac{\partial \!\  \pi}{\partial  R_s}=0$, where $\varsigma \triangleq 2^{R_t}$.
	
	Since $\varsigma$ only appears in the denominator of \eqref{a<b}, by taking the partial derivative of \eqref{a<b} about $\varsigma$, 
	\begin{equation} \label{partial 1‘}
	\frac{\partial \  \pi}{\partial \ \varsigma}=0
	\Leftrightarrow
	\frac
	{\partial
		\left(e^{\frac{\left(\varsigma-1\right)}{\rho_d}}
		\left(1+
		\frac{k_2 \left(\varsigma -1\right)}{\frac{\varsigma}{2^{R_s}}-1}
		\right)\right)
	}
	{\partial \varsigma}=0,
	\end{equation}
	which can be further expanded and simplified as
	\begin{equation} \label{middel_1}
	e^{\frac{\varsigma}{\rho_d}}
	\left(
	\frac{1}{\rho_d}\left(1+\frac{k_2 \left(\varsigma-1\right)}{\frac{\varsigma}{2^{R_s}}-1}\right)
	-\frac{k_2 \left(1-\frac{1}{2^{R_s}}\right)}{\left(\frac{\varsigma}{2^{R_s}}-1\right)^2}
	\right)=0.
	\end{equation}
	Because $e^{\frac{\varsigma}{\rho_d}} >0 $, \eqref{middel_1} is equivalent to
	\begin{equation} \label{middel_2}
	\begin{aligned}
		&\left(\frac{\varsigma}{2^{R_s}}-1\right)\!
		\left(
		\frac{\varsigma}{2^{R_s}}\!-\!1 \!+ \!k_2 2^{R_s}\left(\!\frac{\varsigma}{2^{R_s}}\!-\!1\!\right) \!+\! k_2 2^{R_s} \left(1\!-\!\frac{1}{2^{R_s}}\!\right)\!
		\right) \\
		&- \rho_d k_2 \left(1-\frac{1}{2^{R_s}}\right)
		=0.
	\end{aligned}
	\end{equation}
	By using $\xi \triangleq \frac{\varsigma}{2^{R_s}}-1$, \eqref{middel_2} can be further simplified as
	\begin{equation}
	\xi^2+\frac{k_2 2^{R_s} \left(1-\frac{1}{2^{R_s}}\right)}{1+k_2 2^{R_s}} \xi
	- \frac{\rho_d k_2 \left(1-\frac{1}{2^{R_s}}\right)}{1+k_2 2^{R_s}} = 0,
	\end{equation}
	which has a single positive root as (since $\xi > 0$)
	\begin{equation} \label{y}
	\!\xi\!\!=\!
	\frac{1}{2}\!\!
	\left(\!
	-\frac{k_2 \!\left(2^{R_s} \!-\!1 \right)}{1+k_2 2^{R_s}}
	\!+\!\!\left(\!\!
	\left(\!\!\frac{k_2 \left(2^{R_s} \!-\!1 \right)}{1+k_2 2^{R_s}}\!\!\right)^2
	\!\!\!\!+\!\frac{4 \rho_d k_2 \left( 1 \!-\!\frac{1}{2^{R_s}} \right)}{1+k_2 2^{R_s}}
	\!\right)^{\!\!\!\frac{1}{2}}
	\right)\!.
	\end{equation}
	
	Also we have
	\begin{equation} \label{Middel_3}
	\frac{\partial \  \pi}{\partial \  R_s}\!=\!
	\frac{\!\left(\!1\!+\!\frac{k_2 \left(2^{R_t}-1\right)}{2^{R_t\!-\!R_s}\!-\!1}\!\right)
		\!-\!R_s \left( \frac{\ln 2 \  k_2 \  2^{R_t-R_s} \left(2^{R_t} \!-\!1\right)}{\left(2^{R_t-R_s}-1\right)^2} \right)
	}
	{
		e^{\frac{2\left(2^{R_t}-1\right)}{\rho_d}}
		\left(1+
		\frac{k_2 \left(2^{R_t} -1\right)}{2^{R_t-R_s}-1}
		\right)^2
	}
	=0.
	\end{equation}
	Since the denominator of the middle term of \eqref{Middel_3} is greater than zero, \eqref{Middel_3} reduces to
	\begin{equation} \label{Middel_4}
	k_2
	\left(2^{R_s} +\frac{2^{R_s}-1}{\xi}\right)
	\left(\ln 2 \ R_s -1 + \frac{\ln 2 \ R_s}{\xi}\right)
	=1,
	\end{equation}
	where $k_2$ is defined in \eqref{def_k2}. \par
	Taking \eqref{y} into \eqref{Middel_4}, optimal $R_s$, $R^{\star}_s$ can be obtained easily by linear search, since the left side of \eqref{Middel_4} is monotonically increasing with $R_s$ which can be easily proved. 
	The optimal $R_t$ can be calculated as
	\begin{equation} \label{Middel_6}
	R^{\star}_t = R^{\star}_s + \log_2 \left(1+\xi^{\star}\right),
	\end{equation}
	where $\xi^{\star}$ can be obtained by taking $R^{\star}_s$ into \eqref{y}. \par

	\textbf{Case II}: If optimal $\left(R_t,R_s\right) \in \hat{\mathcal{D}} \cup \mathcal{D}_2$,
	\eqref{NAE_N=1} can be rewritten as
	\begin{equation} \label{a>b}
	\pi= \frac{R_s}{1+\frac{k_1}{2^{R_t-R_s}-1}}.
	\end{equation}
	Because $\pi$ in \eqref{a>b} increases with $R_t$, optimal $R_t$ and $R_s$ should be found at the boundary of $\mathcal{D}_1$ and $\mathcal{D}_2$, that is $\mathcal{\hat{D}}$.
	Letting $\left(a\right) = \left(b\right)$, we have
	\begin{equation} \label{Middel_5}
	1+\frac{k_1}{2^{R_t-R_s}-1} =e^{\frac{2^{R_t}-1}{\rho_d}}
	+ \frac{k_2 \ e^{\frac{2^{R_t}-1}{\rho_d}} \left(2^{R_t} -1\right)}{2^{R_t-R_s}-1},
	\end{equation}
	which can be further simplified as
		\vspace*{-5pt}
	\begin{equation} \label{R_s_temp1}
	2^{R_s} = \frac{2^{R_t}}{1+ \zeta},
	\end{equation}
	where
	\begin{equation} \label{U}
	\zeta= \frac{k_1-k_2 \ e^{\frac{2^{R_t}-1}{\rho_d}} \left(2^{R_t}-1\right)}{e^{\frac{2^{R_t}-1}{\rho_d}} - 1} > 0.
	\end{equation}
	Thus, from \eqref{R_s_temp1} we have
		\vspace*{-5pt}
	\begin{equation} \label{R_s_new}
	R_s =R_t -  \log_2 \left( 1+\zeta \right).
	\end{equation}
	By taking \eqref{R_s_new} into \eqref{a>b}, we have
	\vspace*{-5pt}
	\begin{equation} \label{pi_a=b}
	\pi = \frac{R_t - \log_2 \left(1+\zeta\right)}{1+\frac{k_1}{\zeta}}.
	\end{equation}
	By taking the derivative of $\pi$ about $R_t$ in \eqref{pi_a=b}, optimal $R_t$ should satisfy
	\begin{equation} \label{R_t_equation}
	\frac{
		\left(1\!- \!\frac{1}{\ln 2} \frac{\zeta'}{1+\zeta}\right) \left(1+\frac{k_1}{\zeta}\right)
		\!-\!
		\left(R_t \!-\! \log_2 \left(1+\zeta\right)\right)\left( \!- \frac{k_1}{\zeta^2} \right) \zeta'
	}
	{\left(1+\frac{k_1}{\zeta}\right)^2}\!=\!0,
	\end{equation}
	where
	\vspace*{-10pt}
	\begin{equation} \label{appendix_A_last}
	\zeta' \!\triangleq\! \frac{\mathrm{d} \zeta}{\mathrm{d} R_t} 
	\!=\!
	\frac{\ln 2 \ e^{\frac{2^{R_t}-1}{\rho_d}}}{\!\!\!\left(\!e^{\frac{2^{R_t}\!-\!1}{\rho_d}} \!-\!1\!\right)^2 }
	\!\left(\!
	k_2 \ 2^{R_t} \!\left(\!1\!+\!\frac{1}{\rho_d} \!-\! e^{\frac{2^{R_t}\!-\!1}{\rho_d}}\! \right) \!-\!\frac{k_1\!+\!k_2}{\rho_d}
	\!\right)\!.
	\end{equation}
	And \eqref{R_t_equation} can be further simplified as
	\begin{equation} \label{R_b_eq_new}
	\zeta' \left(
	\frac{1+\frac{k_1}{\zeta}}{\ln 2 \left(1+\zeta\right)}
	- \frac{k_1 \left(R_t - \log_2 \left(1+\zeta\right)\right)}{\zeta^2}
	\right)=1.
	\end{equation}
	
	Thus, $R^{\star}_t$ can be calculated as the solution of \eqref{R_b_eq_new}, and from \eqref{R_s_new}
	\begin{equation}
	R^{\star}_s = R^{\star}_t - \log_2\left(1+\zeta^{\star}\right),
	\end{equation}
	where $\zeta^{\star}$ is calculated by taking $R^{\star}_t$ into \eqref{U}.
	\par
Note that, 
if the optimal $\left(R_t, R_s\right)$ for problem \eqref{a<b} is obtained in region $\mathcal{D}_1$, they are the optimal rate parameters for problem \eqref{opt_new}. 
This is because, firstly, 
the above discussion and derivations show that the optimal rate parameters can only be obtained in region $\mathcal{D}_1$ and $\hat{\mathcal{D}}$.
Secondly,
by using the continuity of the function of throughput \eqref{NAE_N=1}, 
if the optimal $\left(R_t, R_s\right)$ for problem \eqref{a<b} are obtained in region $\mathcal{D}_1$, 
the maximal throughput in region $\mathcal{D}_1$ (i.e., the maximal value of the object function of \eqref{a<b} in $\mathcal{D}_1$), is larger than its boundary $\hat{\mathcal{D}}$. 
Thus, the optimal rate parameters are obtained and fall in region $\mathcal{D}_1$.
	
	\section{Proof of Corollary \ref{C1}}
	We consider the asymptotically high SNR regime, i.e., $\rho_d \rightarrow \infty$ or equivalently $\mathcal{P}_s \rightarrow \infty$. 
	
	When $\rho_d \rightarrow \infty$, we firstly assume $\left(R^{\star}_t,R^{\star}_s\right)$ is obtained in the region $\mathcal{D}_1$.
	The value of $R_s$ that satisfies \eqref{Middel_4_new} cannot go to infinity regardless of the value of $\xi$. Thus, we have $\xi \rightarrow \infty$ as $\rho_d \rightarrow \infty$, and \eqref{Middel_4_new} can be rewritten as
	\vspace*{-5pt}
	\begin{equation} \label{C1_new}
	k_2 2^{R_s} \left(\ln 2 R_s -1\right) =1,
	\end{equation}
	where $k_2$ is defined in \eqref{def_k2}.
	From \eqref{C1_new} optimal $R_s$ for the case $\rho_d \rightarrow \infty$ can be calculated as
	\begin{equation} \label{R_star_new}
	R^{\star}_s = \frac{1 + \mathrm{W}_0 \left(\frac{1}{k_2 2^{\frac{1}{\ln 2}}}\right)}{\ln 2}
	=\frac{1 + \mathrm{W}_0 \left(\frac{1}{e k_2}\right)}{\ln 2}.
	\end{equation}
	From \eqref{y_previous}, we know that $\xi = \mathcal{O} \left(\rho^{\frac{1}{2}}_d\right) = \mathcal{O} \left(\mathcal{P}^{\frac{1}{2}}_s\right)$,
	and because $\xi = \frac{2^{R_t}}{2^{R_s}} - 1$, we have $2^{R_t} =  \mathcal{O} \left(\mathcal{P}^{\frac{1}{2}}_s\right)$. 
	It can be easily verified that the assumption that optimal $\left(R_t,R_s\right) \in \mathcal{D}_1$ is correct. 
	From Proposition \ref{P3} and \eqref{NAE_N=1}, optimal $\left(R_t,R_s\right)$ and $\pi$ is obtained. 
	
	\section{Proof of Proposition \ref{P4}}
	Because $\left(a\right)$ in \eqref{NAE_N=1} decreases with $R_t$, while $\left(b\right)$ increases with $R_t$, 
	optimal $R_t$ can be obtained when the two parts become equal with each other, 
	i.e., optimal $\left(R_t,R_s\right) \in \hat{\mathcal{D}}$.
	Thus, optimization problem \eqref{opt_new} can be rewritten as 
	\vspace*{-0.3cm}
	{\small{
	\begin{equation} \label{p1}
	\begin{aligned}
	&\max_{R_t, R_s}\ \ 
	\frac{R_s}
	{1+\max\left\lbrace
		\frac{d^m_{SJ}}{N_J \eta} \frac{d^m_{JE}}{d^m_{SE}}\frac{\left(N_J-1\right)\left(\varepsilon^{-\frac{1}{N_J-1}}-1\right)}{2^{R_t-R_s}-1},
		e^{\frac{2^{R_t}-1}{\rho_d}}-1
		\right\rbrace}
	\\
    &\text{s.t.}\ 
	\frac{d^m_{SJ}}{N_J \eta} \frac{d^m_{JE}}{d^m_{SE}}\frac{\left(N_J\!-\!1\right)\left(\varepsilon^{-\frac{1}{N_J-1}}\!-\!1\right)}{2^{R_t-R_s}-1}
	=e^{\frac{2^{R_t}-1}{\rho_d}}-1,\ \!\!R_t \!\geq\! R_s \!\geq\! 0.\\	
	\end{aligned}
	\end{equation}}}
	
	By solving the equality constraint, we have
	\vspace*{-0.3cm}
	{\small{
	\begin{equation} \label{Rs}
	2^{R_s}=\frac{2^{R_t}}{1+\frac{M}{e^{\frac{2^{R_t}-1}{\rho_d}}-1}},
	\end{equation}}}
	where $M$ is defined in \eqref{my_M}.
	Certainly, $R_t \geq R_s$ is satisfied in \eqref{Rs}.
	By taking \eqref{Rs} into \eqref{p1},
	the optimization problem can be rewritten as
	{\small{
	\begin{equation} \label{p2}
	\begin{aligned}
	\max_{R_t \geq 0}\ \ \ &
	\frac{\log_2 \left(\frac{2^{R_t}}{1+\frac{M}{e^{\frac{2^{R_t}-1}{\rho_d}}-1}}\right)}
	{ \ e^{\frac{2^{R_t}-1}{\rho_d}}}.\\
	\end{aligned}
	\end{equation}}}
	
	Now we use $z$ to denote $2^{R_t}$, thus $R_t = \log_2 z$.
	By taking the derivative of objective function about $z$ in \eqref{p2}, and then setting it equal to $0$, optimal $z$, $z^{\star}$ can be calculated as the solution of \eqref{equal_previous} which is monotone decreasing function with $z$ on the left side. 

	\vspace*{-0.5cm}
	\section{Proof of Corollaries \ref{C3} and \ref{C2}}
		When $\rho_d \rightarrow \infty$, \eqref{equal_previous} approximates as
		$ 
		2\frac{\rho_d}{z} -\ln z = 0.
		$ 
		Thus, we have
		$ 
		z^{\star} = \frac{2 \rho_d}{\mathrm{W_0}\left(2 \rho_d\right)}.
		$ 
		From \eqref{NAE_N=1} and Proposition \ref{P4}, the Corollary \ref{C3} can be easily obtained.
		When $N_J \rightarrow \infty$, from \eqref{my_M}, we have
		$
		M=
		\frac{d^m_{SJ}}{N_J \eta} \frac{d^m_{JE}}{d^m_{SE}}{\left(N_J-1\right)\left(\varepsilon^{-\frac{1}{N_J-1}}-1\right)} \rightarrow 0.
		$
		Therefore, \eqref{equal_previous} approximates to
		$ 
		\frac{\rho_d}{z}-\ln z = 0.
		$ 
		Thus, we have the expression of optimal $z$ in $N_J \rightarrow \infty$ regime as
		$ 
		z^{\star} = e^{\mathrm{W}_0 \left(\rho_d\right)}.
		$ 
		From \eqref{NAE_N=1} and Proposition \ref{P4}, Corollary \ref{C2} can be easily obtained.
		\vspace*{-0.4cm}
\end{appendices}

\ifCLASSOPTIONcaptionsoff
  \newpage
\fi

\end{document}